\documentclass[12pt]{article}
\usepackage{amsmath}
\usepackage{graphicx}
\usepackage{natbib}
\usepackage{amsfonts,amssymb}
\usepackage{epstopdf}
\usepackage{amsthm}
\usepackage{color}
\usepackage{listings} 
\newcommand{\blind}{0}

\newtheorem{theorem}{Theorem}
\newtheorem{lemma}{Lemma}
\newtheorem{corollary}{Corollary}

\begin{document}

\def\spacingset#1{\renewcommand{\baselinestretch}%
{#1}\small\normalsize} \spacingset{1}


\if0\blind
{
  \title{\bf  The Cauchy Combination Test under Arbitrary Dependence Structures}
  \author{Mingya Long\\
   Academy of Mathematics and Systems Science\\  Chinese Academy of Sciences, University  of Chinese Academy of Sciences\\
    Zhengbang Li\\
    Central China Normal University\\
    Wei Zhang\\
    Academy of Mathematics and Systems Science\\  Chinese Academy of Sciences, University  of Chinese Academy of Sciences\\
    and \\
    Qizhai Li \thanks{
    Corresponding author: liqz@amss.ac.cn. }\hspace{.2cm}\\
    Academy of Mathematics and Systems Science\\  Chinese Academy of Sciences, University  of Chinese Academy of Sciences\\
    \\
    {\emph{ Accepted by The American Statistician.}} }
    \maketitle
} \fi

\if1\blind
{
  \bigskip
  \bigskip
  \begin{center}
    {\LARGE\bf  The Cauchy Combination Test under Arbitrary Dependence Structures}
\end{center}
  \medskip
} \fi

\bigskip
\begin{abstract}
Combining individual $p$-values to perform an overall test is often encountered in statistical applications. The Cauchy combination test (CCT) ({\em J Am Stat Assoc}, 2020, 115(529): 393-402) is a powerful and computationally efficient approach to integrate individual $p$-values under arbitrary dependence structures for sparse signals. We revisit this test to additionally show that (i) the tail probability of the CCT can be approximated just as well when more relaxed assumptions are imposed on individual $p$-values compared to those of the original test statistics;
(ii) such assumptions are satisfied by six popular copula distributions;
and (iii) the power of the CCT is no less than that of the minimum $p$-value test when the number of $p$-values goes to infinity under some regularity conditions.
These findings are confirmed by both simulations and applications in two real datasets, thus further broadening the theory and applications of the CCT.
\end{abstract}

\noindent%

\spacingset{2}

\bigskip

\section{Introduction}
\label{sec:intro}

Combining individual $p$-values to perform an overall test is a long-standing problem in statistics with wide-ranging applications in, for example, genetics, genomics, and functional magnetic resonance imaging. We consider $m$ hypothesis testing problems with a test statistic constructed for each one.
{Let $p_i$ be the $p$-value for the $i$-th hypothesis testing problem}, $i=1, ..., m$.
Here, we cite four well-known conventional approaches for combining $p_1, ..., p_m$,
 $-2\sum_{i=1}^m\ln p_i$  \citep{r13}, 
 $-\sum_{i=1}^m\ln(1-p_i)$  \citep{r28}, 
 $\sum_{i=1}^m\Phi^{-1}(1-p_i)$ \citep{r23},  
 and $\sum_{i=1}^m p_i$  \citep{r12}, where $\Phi(\cdot)$  is the cumulative
distribution function of a standard normal distribution. 
However, these approaches perform well only when signals are dense, i.e., when most $p_i$'s are small.

Recent high-dimensional data collected in multiple disciplines tend to have very sparse signal, with low signal-to-noise ratio.
 Therefore, the application of these conventional methods could result in substantial power loss.
  To overcome this limitation, a number of alternative tests can be applied, including,
 the Tippett's minimum $p$-value test (MINP)  \citep{r29}, 
 which will become important in the present work,
 as well as the Berk-Jones test \citep{r2}, 
 the  higher criticism test   \citep{r11}, 
 the group-combined test \citep{r16}, 
the generalized higher criticism test  \citep{r1}, 
 and the generalized Berk-Jones test \citep{r43}. 
 However, most of these tests do not provide analytic fomulas for calculating $p$-values when individual $p$-values are correlated.
 Resampling approaches, such as permutation and bootstrap procedures, can be used to handle correlated $p$-values. However, such methods are computationally impractical when it comes to large-scale data,  especially when the $p$-value of a combination test is  extremely small.

Recently, Liu and Xie (2020) 
proposed a Cauchy combination test denoted by $\text{CCT}=\sum_{i=1}^m \omega_i\tan\big((0.5-p_i)\pi\big)$, where the weights $\omega_i$ are nonnegative and $\sum_{i=1}^m \omega_i=1$.
They showed that the tail probability of the CCT could be well approximated by a standard Cauchy distribution under the null hypothesis,
which specifies the bivariate normality and two mild assumptions in the high-dimensional setting.
{So far,} the CCT has been used in quite a few {real applications}. For example, Gorfine et al. (2020) 
 applied it to study right-censoring data in survival analysis,
 while McCaw et al. (2020) 
  constructed a powerful test based on the CCT for quantitative trait genetic association studies.
    Khalid et al. (2020) also 
    employed the CCT to perform inter-module communications for runtime hardware Trojan detection.
    The CCT has also been used to analyze whole-genome sequencing studies and genome-wide studies with summary statistics (Li et al., 2019; Liu et al., 2019; Bu et al., 2020; Li et al., 2020; Xu et al., 2020; Li et al., 2021;).
     %

If these assumptions in Liu and Xie (2020) are violated, the approximation of the tail distribution of the CCT with the derived Cauchy distribution might not be appropriate.
Although extensive numerical simulation studies have been conducted in the literature to investigate the feasibility of such approximation,
a theoretical justification is not available.
In this work, we revisit the theoretical underpinnings of the CCT,
 showing that the approximation of the standard Cauchy distribution for the tail probability of the CCT is still valid under a broader range of bivariate distributions, including the six popular copula  distributions. We further show that
the power of the CCT is no less than that of the MINP,
 when the number of tests goes to infinity.
These extensions broaden the theory and applications of the CCT.

This paper is organized as follows. The main results are presented in Section 2. Simulation studies are conducted in Section 3 to examine the accuracy of the tail probability  approximation.
In Section 4, data from prostate cancer and air quality studies are analyzed to further investigate the performance of the CCT and {three conventional} tests. Some discussions are given in the final section, and all technical details are provided in the Supplementary Materials.

\section{ Main Results}

 For each individual $p$-value $p_i$, let $Z_i$ be the corresponding test statistic, $i=1, ..., m$,
  $Z_i$ has zero mean and unit standard deviation under the global null hypothesis.
  Denote by $\rho_{ij}$ the correlation coefficient between $Z_i$ and $Z_j$, and $i,j=1, ..., m.$  Write $\boldsymbol R=(\rho_{ij})_{m\times m}$ and assume that:
\begin{enumerate}
\item[\noindent (C1)] $(Z_i,Z_j)^\top$  follows a bivariate normal distribution for $1 \leq i<j \leq m$, where the superscript $^\top$ denotes the transpose of a matrix or a vector.
\item[\noindent (C2)] { $\lambda_{\max}(\boldsymbol R)<C_0$, where $\lambda_{\max}(\boldsymbol R)$ is the largest eigenvalue of $\boldsymbol R$ and $C_0$ is a positive constant.}
\item[\noindent (C3)] { a constant $\rho_{\max}$ exists such that} $\underset{{1\leq i < j \leq m}}{\max}|\rho_{ij}| \leq \rho_{\max}<1$.
\end{enumerate}

Under assumption (C1) with fixed $m$ or assumptions (C1) to (C3) with $m=o(t^\delta)$, $\delta\in\big(0,1/2\big)$, Liu and Xie (2020)
showed that
\begin{equation}\label{CCT1}
 \lim\limits_{t \rightarrow \infty}\frac{P\left(\text{CCT}>t\right)}{0.5-\arctan(t)/\pi}=1,
\end{equation}
where $o(t^\delta)/(t^\delta)\rightarrow 0 $ as $t \rightarrow \infty$ and the denominator is the tail probability of the standard Cauchy distribution.

{It is worth pointing out that the bivariate normal distribution assumption (C1) for $(Z_i, Z_j)^\top$ can be too stringent for real applications. Instead, $Z_i$ and $Z_j$ could have an arbitrary bivariate distribution. Moreover, assumption (C2) requires the eigenvalues of $\boldsymbol R$ be bounded by a constant, which is not always appropriate in practice}. For example, the largest eigenvalue of the common spiked correlation model can go to infinity (Johnstone, 2001; Lee et al., 2014; Zhang et al., 2020; Shi et al., 2022) (see details of the spiked correlation model in Section 3).

\subsection{ Tail null distribution of the  CCT for fixed $m$}

Instead of assumption (C1), we consider the following less stringent distribution assumption on the individual $p$-values $p_1, ..., p_m$.

\begin{enumerate}
\item[\noindent (D1)]For $1\leq i<j\leq m$, as $t \rightarrow \infty$, there exists a sequence of $\delta_t$ , with $\lim\limits_{t \rightarrow \infty} \delta_t \rightarrow 0$, and the $\lim\limits_{t \rightarrow \infty} \delta_t t \rightarrow \infty$, such that
$$ P\Big( 0 < p_i < \frac{\omega_i}{\pi} \frac{m}{t},~
0<p_j < \frac{\omega_j}{\pi} \frac{m}{\delta_{t}t}
\Big)= o\Big(\frac1{t}\Big)$$
{and }$$  P \Big( 0 < p_i < \frac{\omega_i}{\pi } \frac{m}{(1+\delta_{t})t},~
1-\frac{\omega_j}{\pi} \frac{m}{\delta_{t}t} < p_j < 1
\Big)=o\Big(\frac1{t}\Big).$$

\end{enumerate}

Under assumption (D1), we can show that the standard Cauchy approximation still holds in the following theorem, the proof of which is given in the Supplementary Materials.
\begin{theorem}\label{th1}
The approximation given by (1) still holds if $p_i$ follows the uniform distribution on $[0,1]$ and assumption (D1) is satisfied.
\end{theorem}

Assumption (D1) imposes no restriction on the type of  joint distribution of $Z_i$ and $Z_j$, $i,j=1, ..., m$.
An arbitrary bivariate distribution (including but not limited to the bivariate normal distribution) for $(T_i, T_j)^\top$ is allowed under assumption (D1).
{To make this case, we give the following six bivariate copula functions which are widely used in applications including finance (Genest and Mackay, 1986; Meyer, 2013) and survival analysis (Geerdens et al., 2018):}

1) Product Copula: $$\mathbb{C}(u_i,v_j) = u_i v_j,\;  \; 1\leq i \neq j\leq m.$$

2) Farlie-Gumbel-Morgenstern (FGM) Copula: $$\mathbb{C}(u_i,v_j) = u_i v_j\{1+\theta(1-u_i)(1-v_j)\}, \; \;\theta \in [-1,1], \; 1\leq i \neq j\leq m.$$

3) Cuadras-Aug\'{e} Copula: $$\mathbb{C}(u_i,v_j) = \big\{\min{ (u_i, v_j)}\big\}^{\theta}(u_i v_j)^{1-\theta}, \;  \;\theta \in [0,1], \; 1\leq i \neq j\leq m.$$

4) Normal Copula:\begin{align*}
&\mathbb{C}(u_i,v_j)=\frac{1}{2\pi\sqrt{1-\rho_{ij}}}\int_{-\infty}^{\Phi^{-1}(u_i)}\int_{-\infty}^{\Phi^{-1}(v_j)}\exp\left(-\frac{x^2-2\rho_{ij}xy+y^2}{2(1-\rho^2_{ij})}\right)dxdy,\\ &  \; 1\leq i \neq j\leq m \; , \, \underset{{1\leq i < j \leq m}}{\max}|\rho_{ij}| \leq \rho_{\max}<1.
\end{align*}

5) Ali-Mikhail-Haq (AMH) Copula: $$\mathbb{C}(u_i,v_j) = \frac{u_iv_j}{1-\theta(1-u_i)(1-v_j)}, \; \;\theta \in [-1,1], \; 1\leq i \neq j\leq m.$$

6) Survival Copula: $$\mathbb{C}(u_i,v_j) = {u_iv_j}\exp{(-\theta \ln u_i \ln v_j)},\;  \; \theta \in [0,1], \; 1\leq i \neq j\leq m.$$

Let the joint distribution of $p_i$ and $p_j$ be modelled by one of the six copula functions described above.
We can show the following result, with the proof being given in the Supplementary Materials.

{ \begin{theorem}\label{th2}
Assumption (D1) is satisfied  for the above six types of copula functions. Therefore, the approximation given by (1) still holds under these copula functions.
\end{theorem}}

%

\subsection{Tail null distribution of the CCT for divergent $m$}

Next, in order to relax assumptions (C1) to (C3), we establish the theory of the tail null distribution of the CCT for divergent $m$. The following assumption is needed.

\begin{enumerate}
\item[\noindent (D2)]For $1\leq i<j\leq m$, as $t \rightarrow \infty$, there exists a sequence of $\delta_t$ , with $\lim\limits_{t \rightarrow \infty} \delta_t \rightarrow 0$, and the $\lim\limits_{t \rightarrow \infty} \delta_t t \rightarrow \infty$, such that
$$ \sup_{1\leq i<j\leq m} P\Big( 0 < p_i < \frac{\omega_i}{\pi} \frac{m}{t},~
0<p_j < \frac{\omega_j}{\pi} \frac{m}{\delta_{t}t}
\Big)= o\Big(\frac1{t^{1+\gamma}}\Big)$$
{and }
$$ \sup_{1\leq i<j\leq m} P \Big( 0 < p_i < \frac{\omega_i}{\pi } \frac{m}{(1+\delta_{t})t},~
1-\frac{\omega_j}{\pi} \frac{m}{\delta_{t}t} < p_j < 1
\Big)= o\Big(\frac1{t^{1+\gamma}}\Big),$$
\end{enumerate}
where $0<\gamma\leq1$.

\begin{theorem}\label{th3}
The approximation given by (1) still holds if $p_i$ follows the uniform distribution on $[0,1]$, $m=o(t^{\gamma/2})$, and assumption (D2) is satisfied.
\end{theorem}

The proof of Theorem 3 is given in the Supplementary Materials.
Similar to the case of fixed $m$,
there is no restriction on the joint distribution of $Z_i$ an $Z_j$, as long as assumption (D2) is satisfied. Furthermore, we can show the following theorem.
{ \begin{theorem}\label{th4}
Assumption (D2) is satisfied  for the above six types of copula functions. Therefore, the approximation given by (1) still holds under these copula functions.
\end{theorem}}

{The proof for this theorem is given in the Supplementary Materials.
As noted earlier, according to Liu and Xie (2020), the conclusion in (1) may not hold for a spiked model,
 which violates assumption (C2), since the largest eigenvalue is not a constant.
However, assumption (D2) allows the use of the spiked model.
Therefore, approximate (1) is still valid for the spiked model.
Simulations results shown later can confirm this.}

\subsection{Power comparison between the CCT and the MINP}

 Previously, we stated that the power of the CCT is no less than that of the minimum $p$-value test when the number of $p$-values goes to infinity under some regularity conditions in the introduction. Accordingly, the MINP is equivalent to such maximum $p$-value test (hereinafter denoted as MAX),
 where $\text{MAX}=\max\left\{ Z_1^2, ..., Z_m^2\right\}$.
 {In the following, we compare the performance of the CCT and MAX.} By following the theoretical settings of Donoho and Jin (2004) 
 and Liu and Xie (2020), 
  we assume that
$(Z_1, ...,Z_m)^\top\sim N_m(\boldsymbol \mu,\boldsymbol R)$, where $\boldsymbol \mu=(\mu_1, ...,\mu_m)^\top$.
 Throughout this work, $Z_i$ has had unit standard deviation, which means that the $\boldsymbol R$ is the  correlation matrix.  The hypothesis testing problem is $H_0:\boldsymbol \mu={\boldsymbol 0}_m$ versus $H_1: \boldsymbol \mu\neq{\boldsymbol 0}_m$, where ${\boldsymbol 0}_m$ is the $m$-dimensional vector with all elements being zero. Under $H_1$, denote the index set of nonzero elements (signals) of $ \boldsymbol \mu$ by $\Omega=\{1\leq i \leq m:\mu_i\neq 0\}$,
  and let the total number of nonzero signals be
$\|\Omega\|=m^{\nu}$, where $\|\cdot\|$ is the cardinality of a set and the parameter $0<\nu<1$ measures the sparsity magnitude of signals. Similarly,  denote the number of zero signals by $\|\Omega^c\|=m^{\kappa}$, where  $0<\kappa <1$ and $\Omega^c=\{1, ...,m\}\backslash\Omega$.
 To derive the asymptotic distribution of MAX,  we first need the following assumption.
\begin{enumerate}
  \item[\noindent (D3)]Define  $\varrho_k =\underset{|i-j|\geq k}{\sup}|\rho_{ij}|, ~k=1, ..., m-1.$
Then, $\varrho_1<1~\hbox{and}~\varrho_k(\log k )^{2+d} {\rightarrow}  0~\hbox{for}~d>0~\hbox{as}~{k\rightarrow \infty}$.
\end{enumerate}

\begin{corollary}\label{cl1}
Under assumption (D3)  and for a large enough $t>0$, we have
$$\lim\limits_{m\rightarrow \infty}P\left(\frac{ \hbox{MAX}-\tilde{a}_m - O\big({(\log m)^{-1} }\big)}{\tilde{b}_m}<t \right)
=\exp\big(\exp(-t)\big),$$
where  $\tilde{a}_m=2\log {m}-\{\log{(\log{m})} + \log(4\pi) - \log 4\} + \{\log{(\log{m})}+\log(4\pi)- \log 4\}/(2 \log{m})$, $\tilde{b}_m=2-1/\log{m}$ and $O\big((\log m)^{-1} \big)/(\log m)^{-1} \rightarrow C_1$, $C_1$ is a constant.
\end{corollary}

Deo (1972)
and Pakshirajan and Hebbar (1977)
considered $\max\left\{| Z_1|, ...,|Z_m |\right\},$  and  showed that under assumption (D3),
$$\lim\limits_{m\rightarrow \infty}P\left(\frac{\max\left\{| Z_1|, ...,|Z_m |\right\}-a_m - O\big({(\log m)^{-1} }\big)}{b_m}<t \right)
=\exp\big(\exp(-t)\big),$$
where $a_m=(2\log m)^{1/2}-{(\log \log m +4\pi -4)}/{(8\log m)^{1/2}}$ and $b_m=(2 \log m)^{-1/2}$.
We want to point out that Corollary \ref{cl1} is similar to the conclusion by Deo (1972)
and Pakshirajan and Hebbar's (1977).
However, the proof of Corrollary 1 is different.
Let
$\xi_{1-\alpha}$ and $\eta_{1-\alpha}$ be the  $1-\alpha$ quantile of the standard Cauchy distribution and the standard Gumbell distribution, respectively, i.e., $\xi_{1-\alpha}=\cot(\pi \alpha)$ and $\eta_{1-\alpha}=-\log(\log \alpha)^{-1}$. The asymptotic powers of the CCT and MAX are thus given by
 $$\beta_{\text{CCT}}=P_{H_1}(\text{CCT}>\xi_{1-\alpha})~~\text{and}~~\beta_{\text{MAX}}=P_{H_1}\bigg(\frac{\sqrt{\text{MAX}}-a_m - O\big({(\log m)^{-1} }\big)}{b_m} > \eta_{1-\alpha} \bigg).$$
We have the following result comparing the powers of the two tests.

\begin{theorem}\label{th3}
 Assumption (D3) is satisfied  and $\min\limits_{i=1, ...,m}\omega_i=O(1/m)$. Then as $m \rightarrow \infty$,
$$\beta_{\text{CCT}}\geq \beta_{\text{MAX}}+o(1).$$
\end{theorem}
When $m$ is large enough, Theorem \ref{th3} shows that the asymptotic power of the CCT is no less than that of MAX, which demonstrates the power advantage of the CCT when combining a large number of individual $p$-values.

\section{ Simulation Studies}

\subsection{ Tail probability approximation}

In this section, we conduct simulation studies to evaluate the accuracy of the tail probability approximation based on the standard Cauchy distribution (SCD). Two spiked models with equal and unequal correlation coefficients are considered.
\begin{itemize}
\item Model 1 (Unequal correlation spiked model): $\lambda_i = m/3^{i}$ for $i=1,...,d$ and $\lambda_i= 1 $ for
 $i=d+1, ...,m$, where $d=4,5$ and $6$.
\end{itemize}

  \begin{itemize}
\item Model 2 (Equal correlation spiked model): $\rho_{ij}=\rho$ for $1\leq i\neq j\leq m$ and $\rho_{ii}=1$ for $1\leq i\leq m$, where $\rho=0.2,0.5$ and $0.8$.
    \end{itemize}
\noindent
{Model 2 has a compound symmetry correlation structure, i.e., $(\rho^{Sgn(|i-j|)})_{m\times m}$ with the largest eigenvalue being $m\rho+(1-\rho)$ and all other eigenvalues being $1-\rho$, where $Sgn(|i-j|)$ is a sign function which equals 0 if $i-j=0$ and $1$ otherwise, and $\rho\in(0,1)$.
It is clear that the largest eigenvalue goes to infinity as $m\rightarrow\infty$. }

We consider $m=10$, 50 and 500. The potential test statistics $(Z_1, ...,Z_m)^\top$ are generated from an $m$-dimensional $t$ distribution $t({\bf0}_m, {\boldsymbol R})$ with mean vector ${\bf0}_m$ and correlation matrix ${\boldsymbol R}=\big(\rho_{ij}\big)_{m\times m}$.
{The marginal distributions of $Z_1, ..., Z_m$ are all set to be the univariate $t$ distribution with 10 degrees of freedom.}


\begin{figure}[htbp]
\includegraphics[width=15cm,height=12cm]{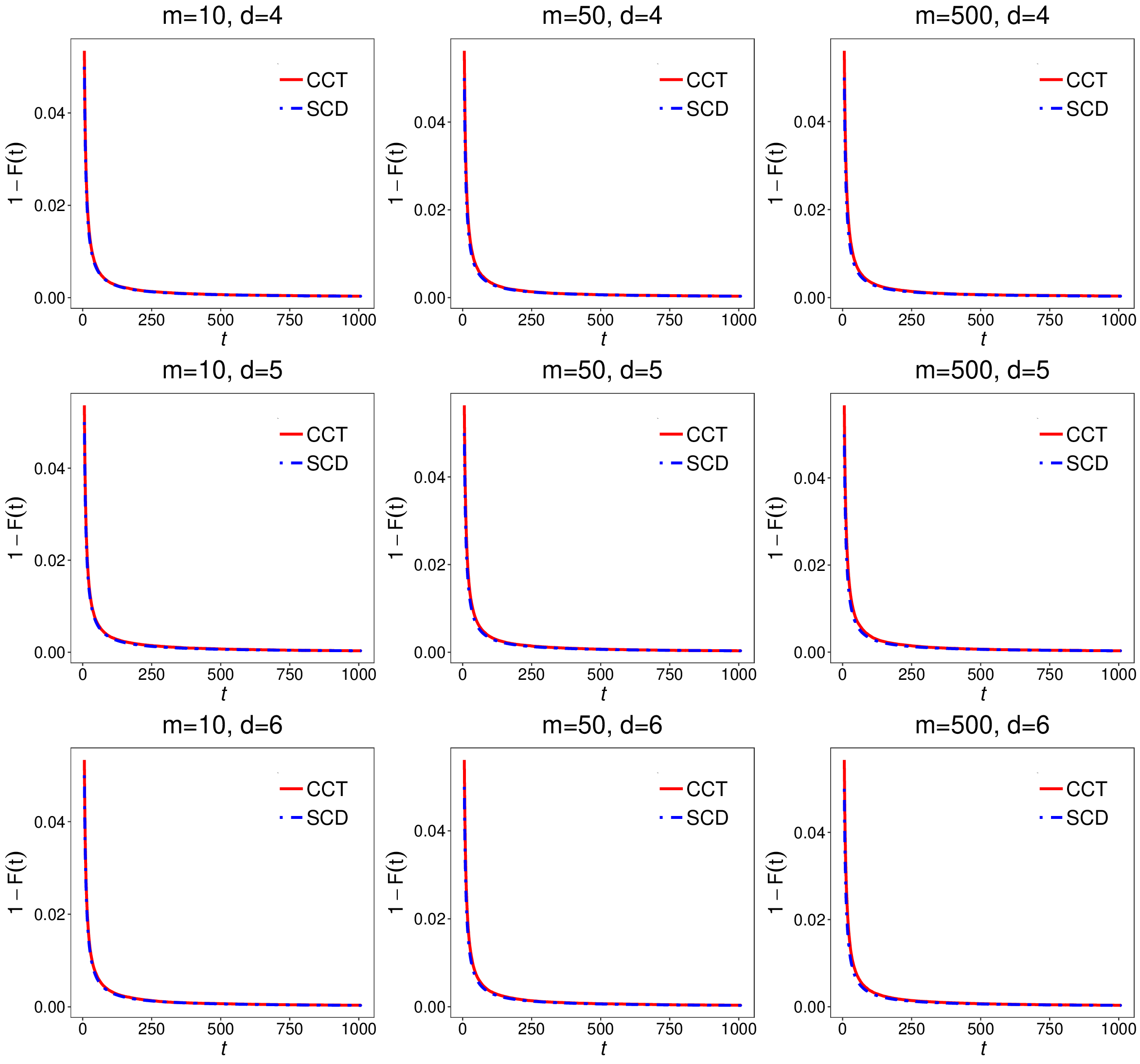}\\
{\small {\bf Figure 1.} { The tail probability of the CCT (red solid line) and the SCD (blue dotted line), where test statistics $Z_{i}, i=1, ...,m$ are generated from
an $m$-variate $t$ distribution, with parameters given in Model 1.
The  vertical axis is the tail probability $1-F(t)$ for a cumulative distribution function $F$}.
}
\end{figure}

\vspace{5mm}

The individual $p$-value is obtained as $p_i=2(1-\Psi(|Z_i|))$, where $\Psi(\cdot)$ is the cumulative distribution function of a $t$ distribution with 10 degrees of freedom.
Figure 1 displays the tail probabilities of $P(\text{SCD}>t)$ (blue dotted line)  and $P(\text{CCT}>t)$ (red solid line)  for Model 1, where  $P(\text{CCT}>t)$ is calculated based on 500,000 Monte Carlo samples. Since the exact distribution of the CCT is unknown, we use this strategy as the gold standard throughout the remainder of this work. In addition, we set  the {range} of the horizontal axis  to be the 95\% and {99.97\%}  quantiles of the standard Cauchy distribution,  which are $6.314$ and 1000, respectively.  From Figure 1, it can be seen  that the standard Cauchy distribution is a good approximation of the CCT, with both lines always coinciding with each other.
The results for Model 2 are similar and details are provided in the Supplementary Materials.


Next, we conduct simulation studies using the {AMH} copula and the FGM copula mixed with the product copula.
 We again consider $m=10$, 50 and 500.
\begin{itemize}

\item Model 3 ({ AMH} copula mixed with product copula model): $(p_i,p_{i+1})^\top\sim\mathbb{C}(u_i,v_{i+1}) = {u_iv_{i+1}}/{\big(1-\theta(1-u_i)(1-v_{i+1})\big)}$, for $i = 1, 3, ..., 2 \lfloor m/2 \rfloor - 1$, where $\lfloor m \rfloor$ {is the maximum integer} less than $m$, where $(p_i,p_{i+1})^\top\sim\mathbb{C}(u_i,v_{i+1}) = u_i v_{i+1}$ for other $i$, and  $\theta=0.2,0.5$ and $0.8$.

\item Model 4 (FGM copula mixed with product copula model): $(p_i,p_{i+1})^\top\sim\mathbb{C}(u_i,v_{i+1}) = u_i v_{i+1}\{1+\theta(1-u_i)(1-v_{i+1})\}$ for $i = 1, 3, ..., 2 \lfloor m/2 \rfloor - 1$, where $(p_i,p_{i+1})^\top\sim\mathbb{C}(u_i,v_{i+1}) = u_i v_{i+1}$ for other $i$, and $\theta=0.2,0.5$ and $0.8$.

\end{itemize}

\begin{figure}[htbp]
\includegraphics[width=15cm,height=12cm]{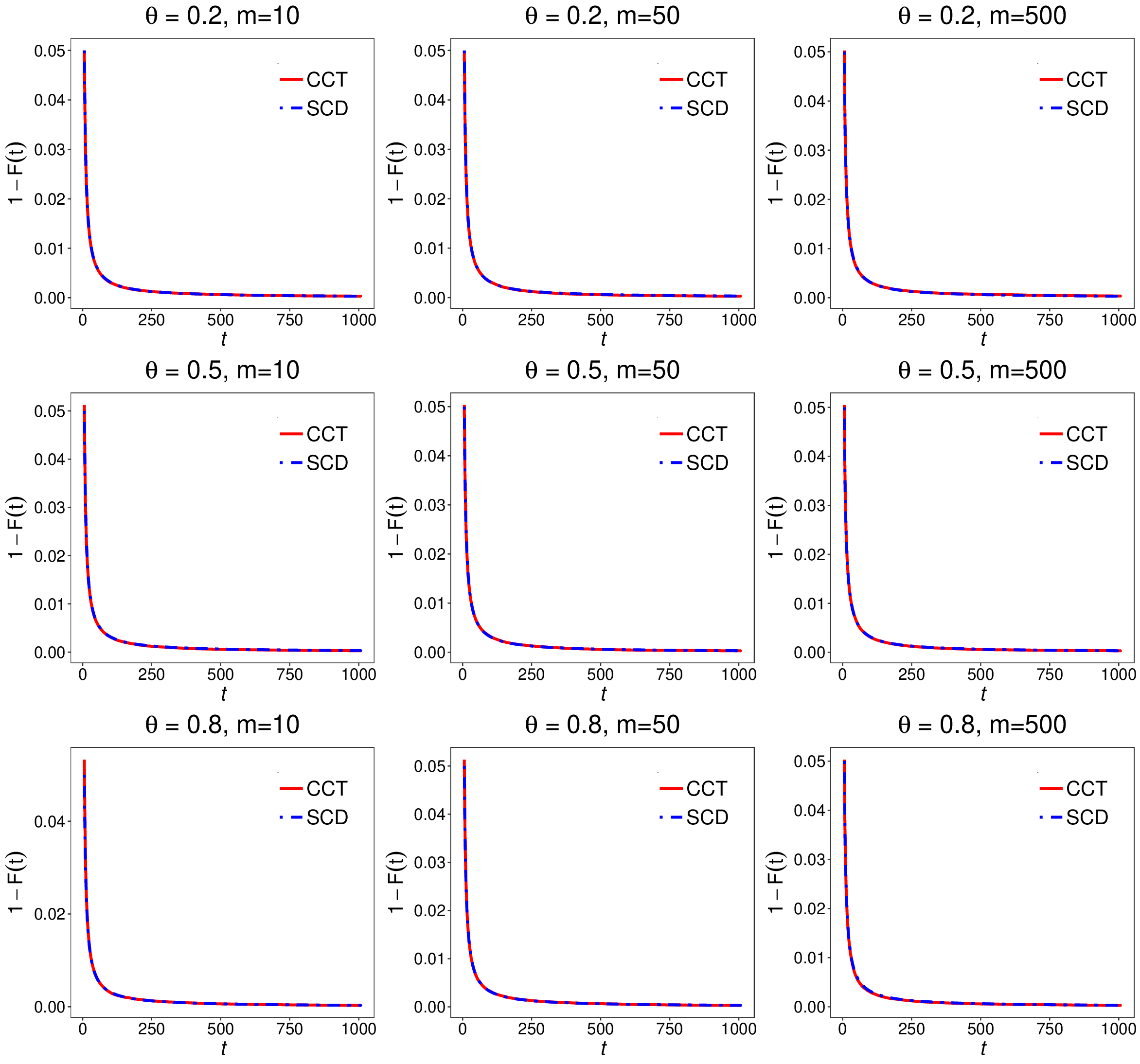}\\
{\small {\bf Figure 2.} { The tail probability of the CCT (red solid line) and the SCD (blue dotted line), where test statistics $Z_{i}, i=1, ...,m$ are generated from
an $m$-variate $t$ distribution, with parameters given in Model 3.
The  vertical axis is the tail probability $1-F(t)$ for a cumulative distribution function $F$}.
}
\label{fig4}
\end{figure}

{The $p$-values are generated based on the above two models.}
Figure 2 shows the tail probabilities of $P(\text{SCD}>t)$ and $P(\text{CCT}>t)$, where $P(\text{CCT}>t)$ is again calculated based on 500,000 {Monte Carlo} samples.
{The two lines in the figure are almost the same, indicating again that the tail probability of the CCT can be well approximated by the standard Cauchy distribution.}
Similar results are observed for Model 4 and are presented in the Supplementary Materials.

\subsection{ Power comparison}

Here, we report a simulation study comparing the power of the CCT to those of MAX,
the generalized higher criticism test (GHC) \citep{r1}
and  generalized Berk-Jones test (GBJ) \citep{r43}.  
{Theorem 5 requires that} $\varrho_1<1$ and $\varrho_k(\log k)^{2+d}{\rightarrow} 0$ for $d>0$ as ${k\rightarrow \infty}$, {which is not satisfied by Models 1 and 2 in Section 3.1. Hence,} we use two  other common models to specify the correlation matrix. The  potential test statistics $(Z_1, ...,Z_m)^\top$ are generated from an $m$-dimensional normal distribution $N({\boldsymbol{\mu}}_m, {\boldsymbol R})$ with mean vector ${\boldsymbol{\mu}}_m$  and correlation matrix ${\boldsymbol R}=\big(\rho_{ij}\big)_{m\times m}$. Matrix $\boldsymbol R$ is specified using the following two structures:
\begin{itemize}
\item Structure 1 (AR(1) correlation): $\rho_{ij}=\rho^{|i-j|}$ for $1\leq i\neq j\leq m$ and $\rho_{ii}=1$ for $1\leq i\leq m$, where $\rho=0.2,0.5$ and $0.8$.

\item Structure 2 (Polynomial decay): $\rho_{ij}=1/{(1+|i-j|^a)}$ for $1\leq i\neq j\leq m$, where $a=0.5,1.5$ and $2.5$.
\end{itemize}
We consider a sparse mean vector ${\boldsymbol{\mu}}_m$ by letting the proportion of {its nonzero elements be} $\|\Omega\|/m = 0.1, 0.2$ and $0.3$,
 {and letting all nonzero elements be equal.}  Thus, the nonzero elements are equal to  $\sqrt{3\log m}/(\|\Omega\|)^{1/3}$, which makes all powers comparable. {The simulation results are calculated based on $n=2,000$ simulation replicates,} when $m$ is chosen from $\{20, 40, 60, 80\}$, and the nominal significance level is $0.05$.

\begin{figure}[htbp]
\includegraphics[width=15cm,height=12cm]{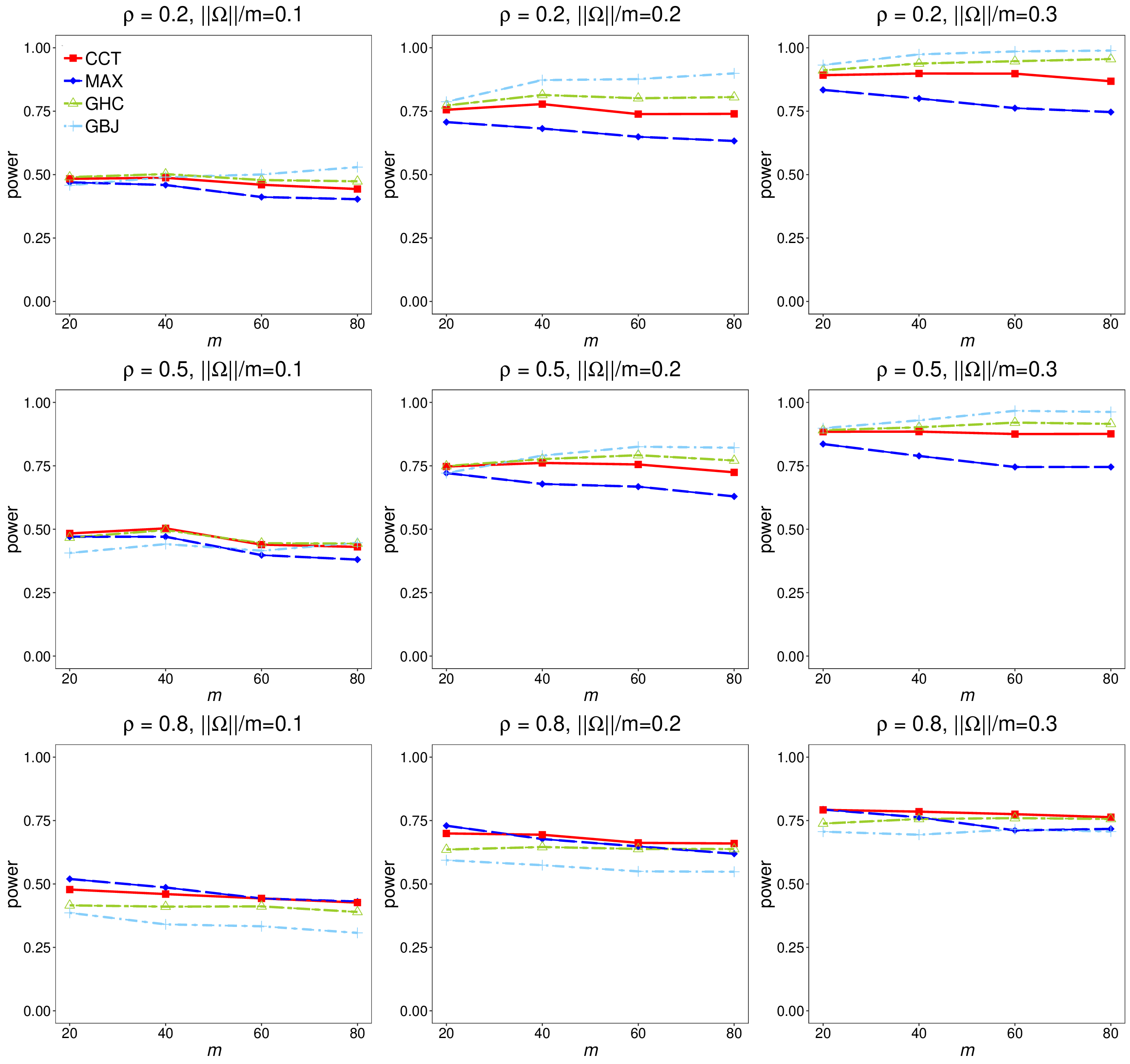}\\
{\small{\bf Figure 3.} {Empirical  powers of CCTs (Squares), MAXes (Circles), GHCs (Triangles), and GBJs (Plus signs), where the test statistics $Z_{i}, i=1, ...,m$ are generated from an $m$-dimensional normal distribution with the parameters set in Structure 1.} }
\label{fig5}
\end{figure}

\begin{figure}[htbp]
\includegraphics[width=15cm,height=12cm]{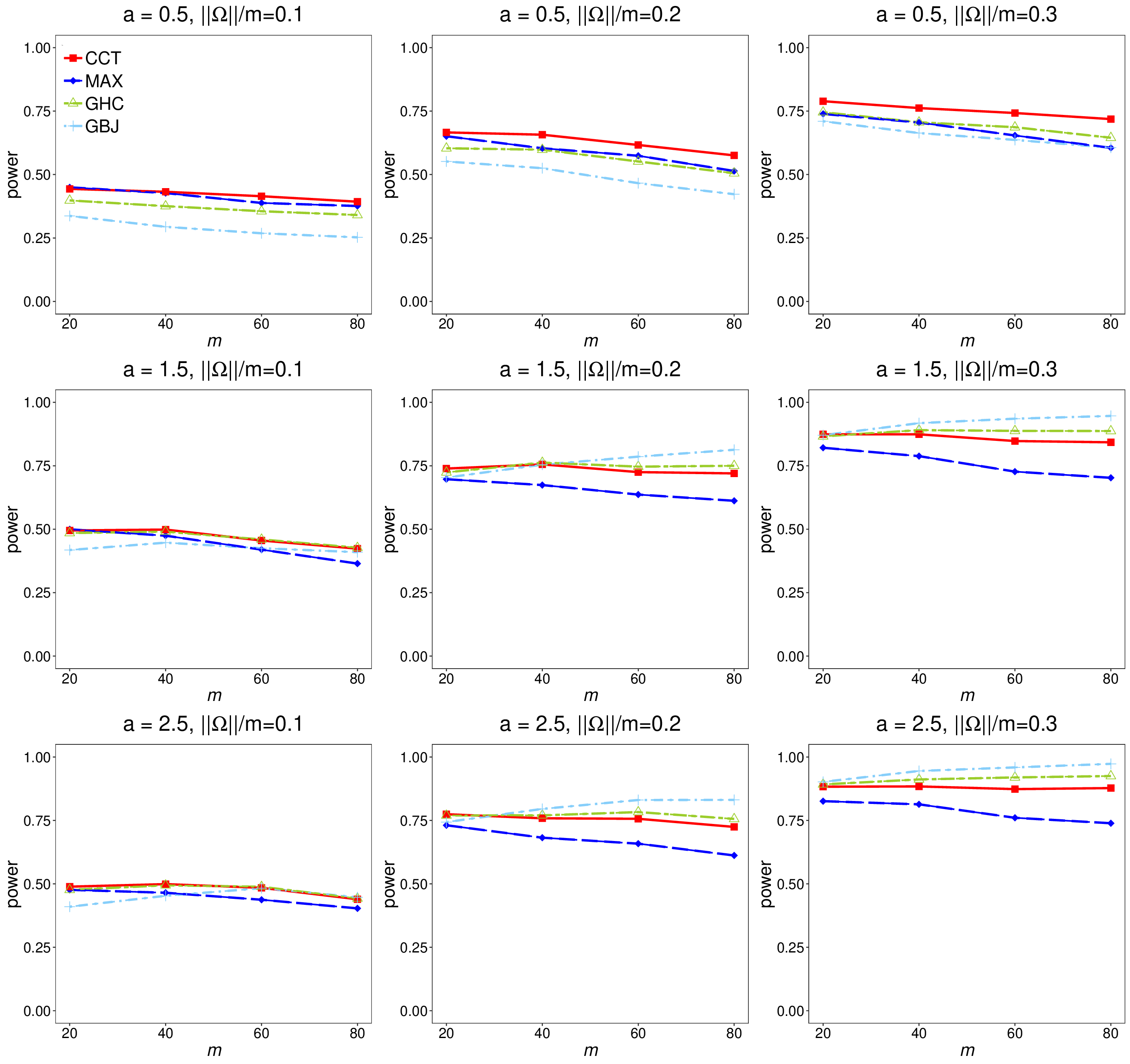}\\
{\small{\bf Figure 4.} { Empirical  powers of CCTs (Squares), MAXes (Circles), GHCs (Triangles), and GBJs (Plus signs), where the test statistics $Z_{i}, i=1, ...,m$ are generated from an $m$-dimensional normal distribution with the parameters set in Structure 2.} }
\label{fig6}
\end{figure}

 Figures 3 and 4 display empirical powers of CCT, MAX, GHC, and GBJ for Structures 1 and 2, respectively.
From these figures, we can see that the CCT has the power comparable to those of GHC and GBJ when individual statistics are weakly correlated. For example, in Figure 3, when $\rho = 0.2$, $\|\Omega\|/m = 0.1$ and $m=60$, the empirical power of the CCT is 0.4600, which is slightly smaller than that of GHC (0.4785) and  GBJ (0.5005). In contrast, when the individual statistics are strongly correlated, the CCT can achieve much higher power. For example, in Figure {4}, when $a = 0.5$, $\|\Omega\|/m = 0.3$ and $m=80$, the empirical power of the CCT is 0.7185, which is much higher than that of GHC (0.6450) and GBJ (0.6060).
The figures further show that the empirical power of MAX is always lower than that of the CCT; in some cases the CCT is about $38\%$ more powerful than MAX.

\section{ Real Data Analyses}

\subsection{A Prostate Cancer Study }
We apply MAX, GHC, GBJ, and CCT to analyze data from a prostate cancer study {aimed to} investigate whether gene expression levels in some pathways are different between tumor and normal prostate samples. The data consist of expression profiles of approximately 12,600 genes from 52 tumor   and 50 non-tumor prostate specimens {\citep{r18}}. 
 Raw data are publicly available at https://www.ncbi.nlm.nih.gov/sites/myncbi/recentactivity.
 Previous literatures show that gene expression levels between tumor and normal samples are different in some pathways,  including
the axon guidance pathway (map04360), 
the toll-like receptor signaling pathway (map04620),  
 the T cell receptor signaling pathway (map04660), 
the caffeine metabolism pathway (map00232),  
the riboflavin metabolism pathway (map00740), 
  and the neuroactive ligand-receptor interaction pathway (map04080), 
  consisting of 163, 148, 182, 10, 12, and 263 genes, respectively
  (Honda and Taniguchi, 2006; Barach et al., 2011; Shtivelman et al., 2014; Eidelman et al., 2017; Ren et al., 2018).
Here, we confine our analyses to these pathways.


The Wilcoxon test is conducted for each {gene to obtain individual $p$-values.}
{However, for a more meaningful comparison, we use the CCT and MAX to combine the individual $p$-values,
and use GHC and GBJ to combine the individual (Wilcoxon) test statistics.}
 Results on these 6 pathways are presented in Table 1.
  The {$p$-values} of MAX, GHC and GBJ are calculated based on 10,000 permutation replicates and the $p$-value for the CCT {is obtained using} the approximation formula (1).
 {Under the nominal significance level of 0.05, Table 1 shows that all four tests, except GBJ, are able to detect significant difference between tumor and normal samples. The GBJ test fails to do so for pathway map04080, with a $p$-value of 0.0514.}
\vspace{5mm}
\begin{table}[htbp]
\label{tab3}
\begin{center}
{\small {\bf Table 1.} $p$-values of MAX, GHC, GBJ, and CCT for 6 pathways of prostate cancer data.  }
\begin{tabular}{p{2.2cm}p{2.2cm}p{2.2cm}p{2.2cm}p{2.2cm}p{2.2cm} }
\hline
Pathway      &\# of genes   &  MAX     & GHC   & GBJ         & CCT       \\\hline
map04360     & ~~163     &0.0004    &0.0003   & 0.0008    & 0.0004    \\
map04620     & ~~148     &0.0145    &0.0152	  & 0.0118    & 0.0078    \\
map04660     & ~~182     &0.0053    &0.0149   & 0.0132    & 0.0079   \\
map00232     & ~~10      & 0.0220    &0.0150	  & 0.0130    & 0.0140    \\
map00740     & ~~12      &0.0167    &0.0188   & 0.0258    & 0.0146    \\
map04080     & ~~263     &0.0232    &0.0341	  & 0.0514    & 0.0117    \\\hline
\end{tabular}\\
\end{center}
\end{table}

\subsection{An Air Quality Study}
{ To further evaluate the performances of  MAX, GHC, GBJ, and CCT, we apply them to analyze data from an air quality study.
The dataset, which is publicly available at http://archive.ics.uci.edu/ml/datasets/Air+quality, contains samples of hourly averaged responses from an air quality chemical multi-sensor device {located in an Italian city}.
Five air pollutants including CO, NOx, NO2, {non-metanic hydrocarbons, and benzene} and three air quality indicators including temperature, relative humidity, and absolute humidity,
were recorded from March 2004 to February 2005 \citep{r55}. 
After deleting the samples with missing data,  827 remaining samples are used for analysis.
Our aim is to check whether the air pollutants and air quality indicators are associated.

Individual $p$-values and individual test statistics are constructed from the $F$ test of overall significance in the regression model for each air pollutant on all air quality indicators.
Each pollutant could be measured two ways. Hence, $m=10$.
We calculate the $p$-values of MAX, GHC and GBJ by using 100,000 permutations and $p$-value of the CCT using the approximation formula (1).
The $p$-values of MAX, GHC, GBJ, and CCT are {$<  10^{-6}$,  $10^{-15}$, $10^{-12}$, and $< 10^{-17}$,} respectively.
Under the nominal significance level of 0.05, all tests could successfully detect the significant association between air pollutants and air quality indicators, consistent with the conclusion of Runge (2018).}

\section{Discussion}

We began by characterizing high-dimensional data signals as sparse, i.e., having only a few small $p$-values. In the performance of an omnibus test, we proposed that individual $p$-values are often encountered in statistical applications. Many conventional approaches to combine $p$-values have been reported in the literatures, such as the MINP, the Berk-Jones test, and the higher criticism test. However, most of these tend to be computationally intensive, limiting their ability. The Cauchy combination test (CCT), on the other hand, is a powerful and computationally efficient approach to integrate individual $p$-values under arbitrary dependence structures for sparse signals. In this work, we revisited the CCT and put forth several assumptions to relax the original ones and, thus, further broaden the theory and applications of the CCT. Specifically, instead of the original test statistics, we first impose assumptions on individual $p$-values. Considering fixed $m$ or divergent $m$, the tail probability of the CCT can be approximated just as well under two assumptions: that 1) as $t$ goes to infinity, the joint probability of any two individual $p$-values is bounded by $1/t$, and 2) the joint probability of any two individual $p$-values is bounded by $1/(t^{1+r})$. Next, in order to confirm the broader application of these two assumptions, six popular and widely used copula distributions are illustrated. Finally, following the theoretical settings of Donoho and Jin (2004) and Liu and Xie (2020), we theoretically prove that the power of the CCT is no less than that of the minimum $p$-value test when the number of $p$-values goes to infinity. These findings are confirmed by both simulation and two real datasets. R codes for the simulation study and real data analyses are given in the Supplemental Materials.

In the case of sparse signals, the analytic formulas for calculating the CCT make the test useful. However, as shown in Chen (2022), the CCT may be less powerful when signals are mixed, i.e., when some $p$-values to be combined are small and others are large. Additionally, Liu and Lin (2018) pointed out that no uniformly powerful test exists for a multiple dimensional composite alternative hypothesis. These challenges call for a novel, robust and widely applicable combination method.
In this work, we have, as noted above, broadened the theoretical basis of the CCT to meet these challenges. However, some improvements are still needed. For instance,  $p$-values do not follow uniform distribution, even though appropriate in some cases. In the case of discrete data analysis, the CCT cannot make an accurate inference when $p$-values do not follow uniform distribution, thus requiring further investigation. As stated in the paper, the approximate strengths of the CCT stem from the dimension $m$, the critical threshold $t$, and the bound of distribution probability of $p$-values, so another interesting direction for future study is to derive the convergence rate of the Cauchy approximation and explore different assumptions to speed it up.

\section*{Aknowledgements}

We would like to thank the editor, the associate editor, and two anonymous reviewers for their insightful comments that have greatly improved the manuscript. We also have to thank David Martin and Kai Yu for their valuable help. Q. Li was supported in part by Beijing Natural Science Foundation (Z180006).

\section*{Conflict of Interest Statement}

The authors report there are no competing interests to declare.

\newpage
\spacingset{1}
  \section*{\bf Supplementary Materials }

{The Supplementary Materials contain three parts. In Part 1, we provide the technical details of theorems and Corollary in Section 2 of the main text. Part 2 includes the additional simulation results for Models 2 and 4 in Section 3. Part 3 gives the R codes for the numerical studies (Model 1 in the simulation study and real data analysis for Air Quality study). }

\section*{Technique Details}
\noindent To prove theorems in the manuscript, we needs three useful Lemmas.

\begin{lemma}\label{le1}
({i}) As $t \rightarrow \infty$,
$$0<0.5-\frac{\arctan(t)}{\pi} < \frac{1}{\pi t}. $$

 \noindent(ii) Let a random variable $U$ follow a Uniform distribution on $[-{\pi}/{2},{\pi}/{2}]$. Then $X=\tan(U)$ follows the standard Cauchy distribution and
$$P\left(X>t\right) = 0.5-\frac{\arctan(t)}{\pi}=\dfrac{1}{\pi t}-o\big(\frac1t\big). $$
\end{lemma}

\begin{proof}[Proof of Lemma \ref{le1}]
({i}) It can be easily shown that $$0.5-\frac{\arctan(t)}{\pi}= \frac{\arctan(1/t)}{\pi}.$$
Since $$\frac{\partial\big((1/t)-\arctan (1/t)\big)}{\partial(1/t)}=1-\frac{1}{1+(1/t)^2}>0 $$ and
$(1/t)-\arctan (1/t)\big|_{t \rightarrow \infty}  \rightarrow 0,$
we have $$\frac1t>\arctan \big(\frac1t\big).$$

({ii}) By Taylor's series expansion, we have
$$\frac{\arctan(1/t)}{\pi}= \frac{1}{\pi t} +  \frac{1}{2}{\frac{\partial^2\big(\arctan (1/t)\big)}{\pi \partial(1/t)^2}\Big|_{1/t=1/t_{\epsilon}}}\big(\frac1t\big)^2,$$
where $1/t_{\epsilon} \in (0,1/t)$ and $ \frac{1}{2}{\frac{\partial^2\big(\arctan (1/t)\big)}{\pi \partial(1/t)^2}\big|_{1/t=1/t_{\epsilon}}}\big(\frac1t\big)^2$ is an error term and {is equal} to $-o(1/t)$.

\end{proof}

\begin{lemma}\label{le2} Let $\Phi(\cdot)$ and $\phi(\cdot)$ be the cumulative distribution function and probability density function of a standard normal distribution, respectively. Then we have

\noindent(i) For any $x>0$,
\begin{center}
	$\Phi^{-1}\left(1-\dfrac{\phi(x)}{x}\right)\leq  x \leq \Phi^{-1}\left(1-\dfrac{\phi(x)}{x} \dfrac{x^{2}}{1+x^{2}}\right).$
\end{center}

\noindent(ii) For $y \rightarrow 0+ $,
\begin{center}
	$ \sqrt{\log(\frac1{y^2})+1} -1 \leq \Phi^{-1}(1-y)\leq \sqrt{\log(\frac1{y^2})}.$
\end{center}

\noindent(iii)  \citep{r7}  
 For a positive integer $m \rightarrow \infty$ and a  positive fixed number $a>0$,

\begin{center}
	$\Phi^{-1}\left(1-\dfrac{a}{m}\right)= \sqrt{2\log (m)} - \frac{\log[\log(m)]+\log(4\pi)}{\sqrt{8\log(m)}}-\frac{\log(a)}{\sqrt{2\log(m)}} +O\left(\frac{1}{\log(m)}\right).$
\end{center}

\noindent(iv) For $y \rightarrow 0+ $,
\begin{center}
	$ \Phi^{-1}(0.5 + y)= {\sqrt{2\pi}}y + o(y) \leq \pi y .$

\end{center}
\end{lemma}

\begin{proof}[Proof of Lemma \ref{le2}]
(i) According to the monotone increasing property of $\Phi(\cdot)$  and the Mill's ratio inequality  \citep{r6},  
 $$\Phi^{-1}\left(1-\dfrac{\phi(x)}{x}\right)\leq  x \leq \Phi^{-1}\left(1-\dfrac{\phi(x)}{x}\times\dfrac{x^{2}}{1+x^{2}}\right).$$

(ii) It can be verified that when $x \rightarrow \infty $ , $x \leq \Phi^{-1}(1-y)$ and
 $$ \log(y) = -\log(\sqrt{2\pi})-\frac{x^2}{2}-\log(x)+\log(\dfrac{x^{2}}{1+x^{2}})  \geq-\frac{x^2}{2}- x,$$
 where $y=\big({\phi(x)}/{x}\big)\times\big({x^{2}}/{(1+x^{2})}\big)$.
 Based on some algebras, we have $$x>\sqrt{\log(\frac1{y^2})+1} -1,$$ implying that $$\sqrt{\log(\frac1{y^2})+1} -1 < \Phi^{-1}(1-y).$$
  Similarly, when $x \rightarrow \infty $ ,  we have $ \Phi^{-1}(1-y)\leq  x$ and
 $$ \log(y) = -\log(\sqrt{2\pi})-\frac{x^2}{2}-\log(x)  \leq -\frac{x^2}{2} ,$$
 where $y={\phi(x)}/{x}$.
So $x<\sqrt{\log1/{y^2}}$ and $ \Phi^{-1}(1-y)< \sqrt{\log1/{y^2}}$.

 (iv) By the Taylor expansion, we have $$ \Phi^{-1}(0.5 + y)= \Phi^{-1}(0.5) + \frac{1}{\sqrt{2\pi}}y + o(y).$$
As $y \rightarrow 0+ $, it follows that $ \Phi^{-1}(0.5 + y)\leq \pi y$.

\end{proof}

\begin{lemma}\label{le3}
 \citep{r3} 
Let $(X,Y)^\top$ be a bivariate normally distributed random variable with $EX=EY=0$, $\hbox{var}(X)=\hbox{var}(Y)=1$ and $\hbox{corr}(X,Y)=\rho$.
 Assume $|\rho|\leq \delta$ for any $\delta$ with $0< \delta <1$, then we have
$$\lim\limits_{c\rightarrow \infty}\frac{P( X>c,~Y>c)}{[2\pi(1-\rho)^{1/2}c^2]^{-1}\exp(-{c^2}{(1+\rho)^{-1}})(1+\rho)^{3/2}}=1.$$
\end{lemma}	

\vspace{5mm}
Lemmas \ref{le1} to \ref{le3} show the bounds of the tail probabilities of the standard Cauchy distribution,  standard normal distribution and
 bivariate normal distribution, respectively.  Now, we move on to prove Theorem 1.
\vspace{5mm}

\begin{proof}[Proof of Theorem 1]
\noindent Three steps are needed to complete the proof of  Theorem 1.

{\em Step 1.} We will decompose $P(\text{CCT}>t)$ into two mutually exclusive events.

Denote
  $$A_{i,t}=\big\{\tan[(0.5-p_i)\pi]>\frac{(1+\delta_{t})t}{\omega_{i}},~\text{CCT}>t\big\}$$
and
$$B_{i,t}=\big\{\tan[(0.5-p_i)\pi]\leq \frac{(1+\delta_{t})t}{\omega_{i}},~\text{CCT}>t\big\},$$
where $\omega_{i}>0$, $1\leq i\leq m$, and $\delta_{t}$ satisfies that $\delta_{t}>0,\delta_{t}\rightarrow0$ and $\delta_{t}t\rightarrow+\infty$ as $t\rightarrow+\infty$.
Let $A_{t}=\bigcup_{i=1}^{m}A_{i,t}$ and $B_{t}=\bigcap_{i=1}^{m}B_{i,t}.$
Then $\{\text{CCT}>t\}=A_{t}\cup B_{t}$. Since $A_{t}$ and $B_{t}$ are mutually exclusive, we have
\begin{center}
   $P(\text{CCT}>t)=P(A_{t})+P(B_{t}).$
\end{center}

{\em Step 2.} we will show that $P(B_{t})=o(1/t)$.

The event $\{\text{CCT}>t\}$ implies that there exists at least one $i$ such that $\tan[(0.5-p_i)\pi]>t/(\omega_{i}m)$. So,
\begin{align*}
  P(B_{t})&\leq \sum\limits_{i=1}^{m}P\big(B_{i,t}\cap\{\tan[(0.5-p_i)\pi]>\frac t{\omega_{i}m}\}\big)\\
  &=\sum\limits_{i=1}^{m}P\left(\frac{t}{\omega_{i}m}<\tan[(0.5-p_i)\pi]\leq \frac{(1+\delta_{t})t}{\omega_{i}},~\text{CCT}>t\right)\\
    &\leq \sum\limits_{i=1}^{m}P\left(\frac{t}{\omega_{i}m}<\tan[(0.5-p_i)\pi]\leq \frac{(1-\delta_{t})t}{\omega_{i}},~\text{CCT}>t\right)\\
    &+\sum\limits_{i=1}^{m}P\left(\frac{(1-\delta_{t})t}{\omega_{i}}<\tan[(0.5-p_i)\pi]\leq \frac{(1+\delta_{t})t}{\omega_{i}}\right)\\
    &\widehat=~I_{1}+I_{2}.
\end{align*}
Note that $\delta_{t}\rightarrow 0$. According to (ii) of {Lemma \ref{le1}}, we have
$$I_{2}= \frac{\omega_i}{(1-\delta_{t})\pi t}-\frac{\omega_i}{(1+\delta_{t})\pi t}+o\big(\frac 1t\big)=o\big(\frac 1t\big)$$
As for $I_{1}$, it can be shown that
\begin{align*}
I_{1} & \leq\sum_{i=1}^{m}P\Bigg(\frac{t}{\omega_{i}m}<\tan[(0.5-p_i)\pi]\leq\frac{(1-\delta_{t})t}{\omega_{i}},~
\sum_{j\neq i}^{m}\omega_{j}\tan[(0.5-p_j)\pi]>\delta_{t}t\Bigg)\\
& \leq\sum_{i=1}^{m}\sum_{j\neq {i}}^{m}P\Bigg\{\frac{t}{\omega_{i}m}<\tan[(0.5-p_i)\pi]\leq\frac{(1-\delta_{t})t}{\omega_{i}},~
\tan[(0.5-p_j)\pi]>\frac{\delta_{t}t}{(m-1)\omega_{j}}\Bigg\}.
\end{align*}
It remains to show for any $1\leq {i}\neq{j}\leq{m}$, 
\begin{align*}
I_{1,ij}&=P\Bigg(\frac{t}{\omega_{i}m}<\tan[(0.5-p_i)\pi]\leq\frac{(1-\delta_{t})t}{\omega_{i}},~
\tan[(0.5-p_j)\pi]>\frac{\delta_{t}t}{(m-1)\omega_{j}}\Bigg)\\
&\leq P\Bigg( \arctan\left( \frac{\omega_{i}}{ (1-\delta_{t})t }\right) \leq p_i <  \frac{1}{\pi} \arctan\left(\frac{\omega_{i}m}{t}\right),~
0<p_j < \frac{1}{\pi} \arctan\left(\frac{(m-1)\omega_{j}}{\delta_{t}t}\right)\Bigg)\\
&\leq P\Bigg( 0 < p_i < \frac{\omega_{i}}{\pi} \frac{m}{t},~
0<p_j < \frac{\omega_{j}}{\pi} \frac{m-1}{\delta_{t}t}
\Bigg)\\
&\leq P\Bigg( 0 < p_i < \frac{\omega_{i}}{\pi} \frac{m}{t},~
0<p_j < \frac{\omega_{j}}{\pi} \frac{m}{\delta_{t}t}
\Bigg)\\
&=o(1/t).
\end{align*}

{\em Step 3.} we will verify $P(A_{t})=1/(t\pi)+o(1/t)$.

 By the Bonferroni inequality \citep{r10},
$$\sum_{i=1}^{m}P(A_{i,t})-\sum_{1\leq{i}<j\leq{m}}P(A_{i,t}\cap A_{j,t})
\leq P(A_t)\leq\sum_{i=1}^{m}P(A_{i,t}).$$
It can be similarly obtained that
$P(A_{i,t}\cap A_{j,t})=o(1/t)$ for any $1\leq{i}<j\leq{m}$. Furthermore,
$$\begin{array}{ccl}
P(A_{i,t})&=& P\big(\tan[(0.5-p_i)\pi]> \frac{(1+\delta_{t})t}{\omega_{i}}\big)
-P\big(\tan[(0.5-p_i)\pi]>\frac{(1+\delta_{t})t}{\omega_{i}},\text{CCT}\leq{t}\big)\\
& \widehat=& J_{1,i} - J_{2,i}.
\end{array}$$
It follows that
$$\begin{array}{ccl}
J_{1,i}&=&P\big(\tan[(0.5-p_i)\pi]>\frac{(1+\delta_{t})t}{\omega_{i}}\big)\\
 &=&\displaystyle\frac{\omega_{i}}{\pi(1+\delta_{t})t}+o(\frac{1}{(1+\delta_{t})t})=\frac{\omega_{i}}{\pi t}+o(1/t)
 \end{array}$$
 and
$$\begin{array}{ccl}
J_{2,i} &=&P\Big(\tan[(0.5-p_i)\pi]>\frac{(1+\delta_{t})t}{\omega_{i}},~\text{CCT}\leq{t}\Big)\\
&\leq& P\Big(\tan[(0.5-p_i)\pi]>\frac{(1+\delta_{t})t}{\omega_{i}},~\sum\limits_{j\neq i}^{m}\omega_{j}\tan[(0.5-p_j)\pi]\leq -\delta_{t}t\Big)\\
& \leq& \sum\limits_{j\neq {i}}^{m}P\Big(\tan[(0.5-p_i)\pi]>\frac{(1+\delta_{t})t}{m\omega_{i}},~
\tan[(0.5-p_j)\pi] \leq -\frac{\delta_{t}t}{(m-1)\omega_{j}}\Big)\\
&=& \sum\limits_{j\neq {i}}^{m} J_{2,ij}.
\end{array}$$

Similar to the derivation of $I_{1,ij},$  we have $J_{2,ij}=o(1/t).$
Hence, $$P(\text{CCT}>t)=\frac{1}{\pi t} + o(\frac{1}{t}).$$
This together with Lemma 1 leads to Theorem 1.
\end{proof}

Next, we will prove that when the joint distribution of $(p_i,p_j)^\top$ follows one of six commonly used copula functions
 with fixed $m$, assumption (D1) is satisfied.
\begin{proof}[Proof of Theorem 2]
1) Product Copula. The copula function is $\mathbb{C}(u_i,v_j) = u_i v_j, 1\leq i\neq j\leq m.$\\
It can be shown that
\begin{align}
&P\Big( 0 < p_i < \frac{\omega_i}{\pi} \frac{m}{t},~0<p_j < \frac{\omega_j}{\pi} \frac{m}{\delta_{t}t}\Big)=\mathbb{C}\Big(\frac{\omega_i}{\pi} \frac{m}{t},\frac{\omega_j}{\pi} \frac{m}{\delta_{t}t} \Big)=\Big(\frac{\omega_i}{\pi} \frac{m}{t}\Big)^2 \frac{1}{\delta_{t}},
\end{align}
and
\begin{align}
&P \Big( 0 < p_i < \frac{\omega_i}{\pi } \frac{m}{(1+\delta_{t})t},~1-\frac{\omega_j}{\pi} \frac{m}{\delta_{t}t} < p_j < 1 \Big)\notag\\
&= \frac{\omega_i}{\pi } \frac{m}{(1+\delta_{t})t} - \mathbb{C}\Big(\frac{\omega_i}{\pi } \frac{m}{(1+\delta_{t})t},1-\frac{\omega_j}{\pi} \frac{m}{\delta_{t}t}\Big)\notag\\
&= \frac{\omega_i}{\pi } \frac{m}{(1+\delta_{t})t} -\frac{\omega_i}{\pi } \frac{m}{(1+\delta_{t})t} \Big(1-\frac{\omega_j}{\pi} \frac{m}{\delta_{t}t}\Big)\notag\\
&= \frac{\omega_i}{\pi } \frac{m}{(1+\delta_{t})t}\frac{\omega_j}{\pi} \frac{m}{\delta_{t}t}.
\end{align}
When $m$ is fixed, from (1), we can obtain that
$$P\Big( 0 < p_i < \frac{\omega_i}{\pi} \frac{m}{t},~0<p_j < \frac{\omega_j}{\pi} \frac{m}{\delta_{t}t}
\Big)=o(\frac1{t}),$$
and from (2), we have
$$P\Big( 0 < p_i < \frac{\omega_i}{\pi} \frac{m}{t},~0<p_j < \frac{\omega_j}{\pi} \frac{m}{\delta_{t}t}
\Big)=o(\frac1{t}).$$
2) Farlie-Gumbel-Morgenstern (FGM) Copula. The copula function is
 $$\mathbb{C}(u_i,v_j) = u_i v_j[1+\theta(1-u_i)(1-v_j)],~ \theta \in [-1,1],~ 1\leq i\neq j\leq m. $$
It can be shown that
\begin{align}
&P\Big( 0 < p_i < \frac{\omega_i}{\pi} \frac{m}{t},~0<p_j < \frac{\omega_j}{\pi} \frac{m}{\delta_{t}t}\Big)=\mathbb{C}\Big(\frac{\omega_i}{\pi} \frac{m}{t},\frac{\omega_j}{\pi} \frac{m}{\delta_{t}t} \Big)\notag\\
&=\Big(\frac{\omega_i}{\pi} \frac{m}{t}\Big)^2 \frac{1}{\delta_{t}}[1+\theta(1-\frac{\omega_i}{\pi} \frac{m}{t})(1-\frac{\omega_j}{\pi} \frac{m}{\delta_{t}t})],
\end{align}
and
\begin{align}
&P \Big( 0 < p_i < \frac{\omega_i}{\pi } \frac{m}{(1+\delta_{t})t},~1-\frac{\omega_j}{\pi} \frac{m}{\delta_{t}t} < p_j < 1 \Big)\notag\\
&= \frac{\omega_i}{\pi } \frac{m}{(1+\delta_{t})t} - \mathbb{C}\Big(\frac{\omega_i}{\pi } \frac{m}{(1+\delta_{t})t},1-\frac{\omega_j}{\pi} \frac{m}{\delta_{t}t}\Big)\notag\\
&= \frac{\omega_i}{\pi } \frac{m}{(1+\delta_{t})t} -\frac{\omega_i}{\pi } \frac{m}{(1+\delta_{t})t} \Big(1-\frac{\omega_j}{\pi} \frac{m}{\delta_{t}t}\Big)\left[1+\theta(1- \frac{\omega_i}{\pi } \frac{m}{(1+\delta_{t})t})\frac{\omega_j}{\pi} \frac{m}{\delta_{t}t}\right]\notag\\
&\leq \frac{\omega_i}{\pi } \frac{m}{(1+\delta_{t})t}\frac{\omega_j}{\pi} \frac{m}{\delta_{t}t}.
\end{align}
When $m$ is fixed, it follows from (3) that
$$P\Big( 0 < p_i < \frac{\omega_i}{\pi} \frac{m}{t},~0<p_j < \frac{\omega_j}{\pi} \frac{m}{\delta_{t}t}
\Big)=o(\frac1{t}),$$
and from (4) that
$$P\Big( 0 < p_i < \frac{\omega_i}{\pi} \frac{m}{t},~1-\frac{\omega_j}{\pi} \frac{m}{\delta_{t}t} < p_j < 1
\Big)=o(\frac1{t}).$$
3) Cuadras-Aug\'{e} Copula. The copula function is
$$\mathbb{C}(u_i,v_j) = \big(\min{ (u_i, v_j)}\big)^{\theta}(u_i v_j)^{1-\theta}, \theta \in [0,1], 1\leq i\neq j\leq m.$$
It can be shown that
\begin{align}
&P\Big( 0 < p_i < \frac{\omega_i}{\pi} \frac{m}{t},~1-\frac{\omega_j}{\pi} \frac{m}{\delta_{t}t} < p_j < 1\Big)=\mathbb{C}\Big(\frac{\omega_i}{\pi} \frac{m}{t},\frac{\omega_j}{\pi} \frac{m}{\delta_{t}t} \Big)=\Big(\frac{\omega_i}{\pi} \frac{m}{t}\Big)^{1-\theta} \frac{\omega_j}{\pi} \frac{m}{\delta_{t}t},
\end{align}
and
\begin{align}
&P \Big( 0 < p_i < \frac{\omega_i}{\pi } \frac{m}{(1+\delta_{t})t},~1-\frac{\omega_j}{\pi} \frac{m}{\delta_{t}t} < p_j < 1 \Big)\notag\\
&= \frac{\omega_i}{\pi } \frac{m}{(1+\delta_{t})t} - \mathbb{C}\Big(\frac{\omega_i}{\pi } \frac{m}{(1+\delta_{t})t},1-\frac{\omega_j}{\pi} \frac{m}{\delta_{t}t}\Big)\notag\\
&= \frac{\omega_i}{\pi } \frac{m}{(1+\delta_{t})t} -\frac{\omega_i}{\pi } \frac{m}{(1+\delta_{t})t} \Big(1-\frac{\omega_j}{\pi} \frac{m}{\delta_{t}t}\Big)^{1-\theta}\notag\\
&\leq \frac{\omega_i}{\pi } \frac{m}{(1+\delta_{t})t}\frac{\omega_j}{\pi} \frac{m}{\delta_{t}t}.
\end{align}
When $m$ is fixed, from (5), we have
$$P\Big( 0 < p_i < \frac{\omega_i}{\pi} \frac{m}{t},~0<p_j < \frac{\omega_j}{\pi} \frac{m}{\delta_{t}t}
\Big)=o(\frac1{t}),$$
and from (6), we have
$$P\Big( 0 < p_i < \frac{\omega_i}{\pi} \frac{m}{t},~1-\frac{\omega_j}{\pi} \frac{m}{\delta_{t}t} < p_j < 1
\Big)=o(\frac1{t}).$$
4) Normal Copula. The copula function is
$$\mathbb{C}(u_i,v_j) = \frac{1}{2\pi\sqrt{1-\rho_{ij}}}\int_{-\infty}^{\Phi^{-1}(u_i)}\int_{-\infty}^{\Phi^{-1}(v_j)}\exp\left(-\frac{x^2-2\rho_{ij}xy+y^2}{2(1-\rho^2_{ij})}\right)dxdy, 1\leq i\neq j\leq m.$$
It can be shown that
\begin{align*}
&P\Big( 0 < p_i < \frac{\omega_i}{\pi} \frac{m}{t},~0<p_j < \frac{\omega_j}{\pi} \frac{m}{\delta_{t}t}\Big)\\
&=\mathbb{C}\Big(\frac{\omega_i}{\pi} \frac{m}{t},\frac{\omega_j}{\pi} \frac{m}{\delta_{t}t} \Big)\\
&=\frac{1}{2\pi\sqrt{1-\rho_{ij}}}\int_{-\infty}^{\Phi^{-1}(\frac{\omega_i}{\pi} \frac{m}{t})}\int_{-\infty}^{\Phi^{-1}(\frac{\omega_j}{\pi} \frac{m}{\delta_{t}t} )}\exp\left(-\frac{x^2-2\rho_{ij}xy+y^2}{2(1-\rho^2_{ij})}\right)dxdy\\
&=\frac{1}{2\pi\sqrt{1-\rho_{ij}}}\int^{\infty}_{\Phi^{-1}(1-\frac{\omega_i}{\pi} \frac{m}{t})}\int^{\infty}_{\Phi^{-1}(1-\frac{\omega_j}{\pi} \frac{m}{\delta_{t}t} )}\exp\left(-\frac{x^2-2\rho_{ij}xy+y^2}{2(1-\rho^2_{ij})}\right)dxdy.
\end{align*}
Assume that $d_0\gg0$ satisfies the following equation
$$\frac{\exp\{-d_0^2/2\}}{d_0\sqrt{2\pi}}\frac{d_0^2}{1+d_0^2}=\sup_{i\in 1, ...,m} \frac{\omega_i m}{\pi \delta_{t} t }.$$
Through some algebras, we can obtain that $d_0 \rightarrow \infty$ as $t \rightarrow \infty$.
Then $$ d_0^2 = O\Big(\log\frac{\pi (\delta_{t} t)^2}{(\omega m)^2}\Big), \quad t \rightarrow \infty,$$
where $ \omega =\arg\sup\limits_{i\in 1, ...,m} {(\omega_i m)}/{(\pi \delta_{t} t) }.$
According to (i) in Lemma \ref{le2} and Lemma \ref{le3}, we can obtain that $\Phi^{-1}(1-{(\omega_i m)}/{(\pi t) })>d_0$, $\Phi^{-1}(1-{(\omega_j m)}/{(\pi \delta_t t) })>d_0$ and
\begin{align}
&P\Big( 0 < p_i < \frac{\omega_i}{\pi} \frac{m}{t},~0<p_j < \frac{\omega_j}{\pi} \frac{m}{\delta_{t}t}\Big)\notag\\
&\leq\frac{1}{2\pi\sqrt{1-\rho_{ij}}}\int^{\infty}_{d_0}\int^{\infty}_{d_0}\exp\left(-\frac{x^2-2\rho_{ij}xy+y^2}{2(1-\rho^2_{ij})}\right)dxdy\notag\\
&=O\left(\frac{1}{ (\delta_{t} t)^{^{\frac{2}{1+\rho_{ij}}}}\log(2\pi {(\delta_{t} t)}^2/m^2)}\right),
\end{align}
and
\begin{align}
&P \Big( 0 < p_i < \frac{\omega_i}{\pi } \frac{m}{(1+\delta_{t})t},~1-\frac{\omega_j}{\pi} \frac{m}{\delta_{t}t} < p_j < 1 \Big)\notag\\
&=\frac{1}{2\pi\sqrt{1-\rho_{ij}}}\int_{-\infty}^{\Phi^{-1}(\frac{\omega_i}{\pi} \frac{m}{t})}\int^{\infty}_{\Phi^{-1}(1-\frac{\omega_j}{\pi} \frac{m}{\delta_{t}t} )}\exp\left(-\frac{x^2-2\rho_{ij}xy+y^2}{2(1-\rho^2_{ij})}\right)dxdy\notag\\
&=\frac{1}{2\pi\sqrt{1-\rho_{ij}}}\int^{\infty}_{\Phi^{-1}(1-\frac{\omega_i}{\pi} \frac{m}{t})}\int^{\infty}_{\Phi^{-1}(1-\frac{\omega_j}{\pi} \frac{m}{\delta_{t}t} )}\exp\left(-\frac{x^2-2\rho_{ij}xy+y^2}{2(1-\rho^2_{ij})}\right)dxdy.
\end{align}
When $m$ is fixed and $0\leq\rho_{ij}<\rho_{\max}<1$, it follows from (7) that
$$P\Big( 0 < p_i < \frac{\omega_i}{\pi} \frac{m}{t},~0<p_j < \frac{\omega_j}{\pi} \frac{m}{\delta_{t}t}
\Big)=o(\frac1{t}),$$
and from (8) that
$$P\Big( 0 < p_i < \frac{\omega_i}{\pi} \frac{m}{t},~0<p_j < \frac{\omega_j}{\pi} \frac{m}{\delta_{t}t}
\Big)=o(\frac1{t}).$$
5) Ali-Mikhail-Haq (AMH) Copula.
(i) When $\theta \in [0,1]$, the copula function is $$\mathbb{C}(u_i,v_j) = \frac{u_iv_j}{1-\theta(1-u_i)(1-v_j)},~1\leq i\neq j\leq m.$$
We can show that
\begin{align}
&P\Big( 0 < p_i < \frac{\omega_i}{\pi} \frac{m}{t},~0<p_j < \frac{\omega_j}{\pi} \frac{m}{\delta_{t}t}\Big)\notag\\
&=\mathbb{C}\Big(\frac{\omega_i}{\pi} \frac{m}{t},\frac{\omega_j}{\pi} \frac{m}{\delta_{t}t} \Big)\notag\\
&\leq \Big(\frac{\omega_i}{\pi} \frac{m}{t}\Big)^2 \frac{1}{\delta_{t}(1-\theta)},
\end{align}
and
\begin{align}
&P \Big( 0 < p_i < \frac{\omega_i}{\pi } \frac{m}{(1+\delta_{t})t},~1-\frac{\omega_j}{\pi} \frac{m}{\delta_{t}t} < p_j < 1 \Big)\notag\\
&= \frac{\omega_i}{\pi } \frac{m}{(1+\delta_{t})t} - \mathbb{C}\Big(\frac{\omega_i}{\pi } \frac{m}{(1+\delta_{t})t},1-\frac{\omega_j}{\pi} \frac{m}{\delta_{t}t}\Big)\notag\\
&\leq \frac{\omega_i}{\pi } \frac{m}{(1+\delta_{t})t} -\frac{\omega_i}{\pi } \frac{m}{(1+\delta_{t})t} \Big(1-\frac{\omega_j}{\pi} \frac{m}{\delta_{t}t}\Big)\notag\\
&= \frac{\omega_i}{\pi } \frac{m}{(1+\delta_{t})t}\frac{\omega_j}{\pi} \frac{m}{\delta_{t}t}.
\end{align}
When $m$ is fixed, from (9), we have
$$P\Big( 0 < p_i < \frac{\omega_i}{\pi} \frac{m}{t},~0<p_j < \frac{\omega_j}{\pi} \frac{m}{\delta_{t}t}
\Big)=o(\frac1{t}),$$
and from (10), we have
$$P\Big( 0 < p_i < \frac{\omega_i}{\pi} \frac{m}{t},~0<p_j < \frac{\omega_j}{\pi} \frac{m}{\delta_{t}t}
\Big)=o(\frac1{t}).$$

\noindent(ii) When $\theta \in [-1,0]$, the copula function is
 $$\mathbb{C}(u_i,v_j) = \frac{u_iv_j}{1-\theta(1-u_i)(1-v_j)},~ 1\leq i\neq j\leq m. $$\\
It is easy to obtain that
\begin{align}
&P\Big( 0 < p_i < \frac{\omega_i}{\pi} \frac{m}{t},~0<p_j < \frac{\omega_j}{\pi} \frac{m}{\delta_{t}t}\Big)=\mathbb{C}\Big(\frac{\omega_i}{\pi} \frac{m}{t},\frac{\omega_j}{\pi} \frac{m}{\delta_{t}t} \Big)\leq \Big(\frac{\omega_i}{\pi} \frac{m}{t}\Big)^2 \frac{1}{\delta_{t}},
\end{align}
and
\begin{align}
&P \Big( 0 < p_i < \frac{\omega_i}{\pi } \frac{m}{(1+\delta_{t})t},~1-\frac{\omega_j}{\pi} \frac{m}{\delta_{t}t} < p_j < 1 \Big)\notag\\
&= \frac{\omega_i}{\pi } \frac{m}{(1+\delta_{t})t} - \mathbb{C}\Big(\frac{\omega_i}{\pi } \frac{m}{(1+\delta_{t})t},1-\frac{\omega_j}{\pi} \frac{m}{\delta_{t}t}\Big)\notag\\
&\leq \frac{\omega_i}{\pi } \frac{m}{(1+\delta_{t})t} -\frac{\omega_i}{2\pi } \frac{m}{(1+\delta_{t})t} \Big(1-\frac{\omega_j}{\pi} \frac{m}{\delta_{t}t}\Big)\notag\\
&= \frac{\omega_i}{\pi } \frac{m}{(1+\delta_{t})t}\frac{\omega_j}{\pi} \frac{m}{\delta_{t}t}- \frac{\omega_i}{2\pi } \frac{m}{(1+\delta_{t})t}.
\end{align}
When $m$ is fixed, from (11), we have
$$P\Big( 0 < p_i < \frac{\omega_i}{\pi} \frac{m}{t},~0<p_j < \frac{\omega_j}{\pi} \frac{m}{\delta_{t}t}
\Big)=o(\frac1{t}),$$
and from (12), we have
$$P\Big( 0 < p_i < \frac{\omega_i}{\pi} \frac{m}{t},~0<p_j < \frac{\omega_j}{\pi} \frac{m}{\delta_{t}t}
\Big)=o(\frac1{t}).$$
6) Survival Copula. The copula function is
$$\mathbb{C}(u_i,v_j) = {u_iv_j}\exp{(-\theta \ln u_i \ln v_j)},~ \theta \in [0,1],~ 1\leq i\neq j\leq m. $$
We can show that
\begin{align}
&P\Big( 0 < p_i < \frac{\omega_i}{\pi} \frac{m}{t},~0<p_j < \frac{\omega_j}{\pi} \frac{m}{\delta_{t}t}\Big)=\mathbb{C}\Big(\frac{\omega_i}{\pi} \frac{m}{t},\frac{\omega_j}{\pi} \frac{m}{\delta_{t}t} \Big)\leq \Big(\frac{\omega_i}{\pi} \frac{m}{t}\Big)^2 \frac{1}{\delta_{t}},
\end{align}
and
\begin{align}
P \Big(&0 < p_i < \frac{\omega_i}{\pi }  \frac{m}{(1+\delta_{t})t},~1-\frac{\omega_j}{\pi} \frac{m}{\delta_{t}t} < p_j < 1 \Big)\notag\\
&= \frac{\omega_i}{\pi } \frac{m}{(1+\delta_{t})t} - \mathbb{C}\Big(\frac{\omega_i}{\pi } \frac{m}{(1+\delta_{t})t},1-\frac{\omega_j}{\pi} \frac{m}{\delta_{t}t}\Big)\notag\\
&= \frac{\omega_i}{\pi } \frac{m}{(1+\delta_{t})t} -\frac{\omega_i}{\pi } \frac{m}{(1+\delta_{t})t} \Big(1-\frac{\omega_j}{\pi} \frac{m}{\delta_{t}t}\Big)\exp{\Big(-\theta \ln \big(\frac{\omega_i}{\pi } \frac{m}{(1+\delta_{t})t}\big) \ln \big(1-\frac{\omega_j}{\pi} \frac{m}{\delta_{t}t}\big) \Big)}\notag\\
&= \frac{\omega_i}{\pi } \frac{m}{(1+\delta_{t})t} -\frac{\omega_i}{\pi } \frac{m}{(1+\delta_{t})t} \Big(1-\frac{\omega_j}{\pi} \frac{m}{\delta_{t}t}\Big)\Big(1+\theta \ln \big(\frac{\omega_i}{\pi } \frac{m}{(1+\delta_{t})t}\big) \ln \big(1-\frac{\omega_j}{\pi} \frac{m}{\delta_{t}t}\big) \Big)\notag\\
&= \frac{\omega_i\omega_j m^2}{\pi^2 (1+\delta_{t})\delta_{t} t^2} \Big(1+\theta \ln \big(\frac{\omega_i m}{\pi (1+\delta_{t})t }\big) \ln \big(1-\frac{\omega_j m}{\pi\delta_{t}t} \big) \Big)\notag\\
& \quad - \frac{\omega_i m \theta}{\pi(1+\delta_{t})t} \ln \big(\frac{\omega_i m}{\pi (1+\delta_{t})t }\big) \ln \big(1-\frac{\omega_j m}{\pi\delta_{t}t} \big)
\end{align}
When $m$ is fixed, from (13), we have
$$P\Big( 0 < p_i < \frac{\omega_i}{\pi} \frac{m}{t},~0<p_j < \frac{\omega_j}{\pi} \frac{m}{\delta_{t}t}
\Big)=o(\frac1{t}).$$
and from (14), we have
$$P\Big( 0 < p_i < \frac{\omega_i}{\pi} \frac{m}{t},~0<p_j < \frac{\omega_j}{\pi} \frac{m}{\delta_{t}t}
\Big)=o(\frac1{t}).$$

\end{proof}

Next we will prove Theorem 3.
\begin{proof}[Proof of Theorem 3]
\noindent Similar to the proof of Theorem 1,  we  decompose $\{\text{CCT}>t\}$  into two mutually exclusive events $A_{t}$ and $B_{t}$, and calculate $P(A_t)$ and $P(B_t)$.

Firstly we will show $P(B_{t})=o(1/t)$.
Denote $$I_1 = \sum\limits_{i=1}^{m}P\left(\frac{t}{\omega_{i}m}<\tan[(0.5-p_i)\pi]\leq \frac{(1-\delta_{t})t}{\omega_{i}},~\text{CCT}>t\right)$$
 and $$I_2 = \sum\limits_{i=1}^{m}P\left(\frac{(1-\delta_{t})t}{\omega_{i}}<\tan[(0.5-p_i)\pi]\leq \frac{(1+\delta_{t})t}{\omega_{i}}\right).$$
Then $  P(B_{t}) \leq I_{1}+I_{2}$.
By noting that $\delta_{t}\rightarrow 0$ and using (ii) of {Lemma \ref{le2}}, we have
$$I_{2}= \frac{\omega_i}{(1-\delta_{t})\pi t}-\frac{\omega_i}{(1+\delta_{t})\pi t}+o(\frac1t)=o(\frac1t).$$
Denote $$I_{1,ij}=P\Big(\frac{t}{\omega_{i}m}<\tan[(0.5-p_i)\pi]\leq\frac{(1-\delta_{t})t}{\omega_{i}},~
\tan[(0.5-p_j)\pi]>\frac{\delta_{t}t}{(m-1)\omega_{j}}\Big),$$
we have
$$I_{1}  \leq \sum_{i=1}^{m}\sum_{j\neq {i}}^{m} I_{1,ij}.
$$
It remains to show for any $1\leq {i}\neq{j}\leq{m}$,~ $I_{1,ij}=o(1/{t^{1+\gamma}}).$
\begin{align*}
I_{1,ij}
&\leq P\Bigg( \arctan\left( \frac{\omega_{i}}{ (1-\delta_{t})t }\right) \leq p_i <  \frac{1}{\pi} \arctan\left(\frac{\omega_{i}m}{t}\right),~
0<p_j < \frac{1}{\pi} \arctan\left(\frac{(m-1)\omega_{j}}{\delta_{t}t}\right)\Bigg)\\
&\leq P\Bigg( 0 < p_i < \frac{\omega_i}{\pi} \frac{m}{t},~
0<p_j < \frac{\omega_j}{\pi} \frac{m-1}{\delta_{t}t}
\Bigg)\\
&\leq P\Bigg( 0 < p_i < \frac{\omega_i}{\pi} \frac{m}{t},~
0<p_j < \frac{\omega_j}{\pi} \frac{m}{\delta_{t}t}
\Bigg)\\
&=o\big(\frac1{t^{1+\gamma}}\big).
\end{align*}

Next, we will verify that $P(A_{t})=1/(t\pi)+o(1/t)$.

By the Bonferroni inequality  \citep{r10}, 
$$\sum_{i=1}^{m}P(A_{i,t})-\sum_{1\leq{i}<j\leq{m}}P(A_{i,t}\cap A_{j,t})
\leq P(A_t)\leq\sum_{i=1}^{m}P(A_{i,t}).$$
In this situation, we have to prove that $\sum\limits_{1\leq{i}<j\leq{m}}P(A_{i,t}\cap A_{j,t})=o(1/t)$.
\begin{align*}
P(A_{i,t}\cap A_{j,t})
&<P\Big(\tan[(0.5-p_i)\pi]>\frac{(1+\delta_{t})t}{m \omega_i},~\tan[(0.5-p_j)\pi]> \frac{(1+\delta_{t})t}{\omega_j}\Big)\\
&\leq
P\Bigg(
0<p_i < \frac{1}{\pi} \arctan
\left(
\frac{m \omega_i}{(1+\delta_{t})t}
\right),~
0<p_j < \frac{1}{\pi} \arctan
\left(
\frac{m \omega_j}{(1+\delta_{t})t}
\right)
\Bigg)\\
&\leq P\Bigg( 0 < p_i < \frac{\omega_i}{\pi} \frac{m}{t},~
0<p_j < \frac{\omega_j}{\pi} \frac{m}{\delta_{t}t}
\Bigg)\\
&=\big(\frac1{t^{1+\gamma}}\big).
\end{align*}
\noindent On the other hand,
$$\begin{array}{ccl}
P(A_{i,t})&=&P\big(\tan[(0.5-p_i)\pi]>\frac{(1+\delta_{t})t}{\omega_{i}}\big)
-P\big(\tan[(0.5-p_i)\pi]>\frac{(1+\delta_{t})t}{\omega_{i}},~\text{CCT}\leq{t}\big)\\
& \widehat= & J_{1,i} + J_{2,i},
\end{array}
$$
where
$$\begin{array}{ccl}
J_{1,i} & = & P\big(\tan[(0.5-p_i)\pi]>\frac{(1+\delta_{t})t}{\omega_{i}}\big)\\
&=&\frac{\omega_{i}}{\pi(1+\delta_{t})t}
+o(\frac{\omega_{i}}{(1+\delta_{t})t})\\
&=&\frac{\omega_{i}}{\pi t}+o(\frac{\omega_{i}}{t})
\end{array}$$
and
$$\begin{array}{ccl}
J_{2,i}& = & P\big(\tan[(0.5-p_i)\pi]>\frac{(1+\delta_{t})t}{\omega_{i}},~\text{CCT}\leq{t}\big)\\
 &\leq& P\bigg(\tan[(0.5-p_i)\pi]>\frac{(1+\delta_{t})t}{\omega_{i}m},~
\sum_{j\neq i}^{m}\omega_{j}\tan[(0.5-p_j)\pi]\leq -\delta_{t}t\bigg)\\
& \leq& \sum\limits_{j\neq {i}}^{m}P\Bigg(\tan[(0.5-p_i)\pi]>\frac{(1+\delta_{t})t}{\omega_{i}m},~
\tan[(0.5-p_j)\pi] \leq -\frac{\delta_{t}t}{(m-1)\omega_{j}}\Bigg)\\
&=& \sum\limits_{j\neq {i}}^{m} J_{2,ij}.
\end{array}$$

So  $P(B_{t})=o(1/t)$.
Similar to $I_{2,ij}$, $J_{2,ij}=o(1/t^{1+\gamma}),$ $1\leq {i}\neq{j}\leq{m}$,
Theorem 2 holds.

\end{proof}
\setcounter{equation}{1}

{
Then, we will prove that when the joint distribution of $(p_i,p_j)^\top$ follows one of  six commonly used copula functions
 with divergent $m$,    assumption (D2) is satisfied.

\begin{proof}[Proof of Theorem 4]
 1) Product Copula. The copula function is $\mathbb{C}(u_i,v_j) = u_i v_j, 1\leq i\neq j\leq m.$

When $m$ is divergent and $m=o(t^{\gamma/2})$, let $\delta_{t}=t^{\gamma-1}$, where $\gamma \in (0,1]$.
From (1), we have
$$P\Big( 0 < p_i < \frac{\omega_i}{\pi} \frac{m}{t},~0<p_j < \frac{\omega_j}{\pi} \frac{m}{\delta_{t}t}
\Big)=o\big(\frac1{t^{1+\gamma}}\big),$$
and from (2), we can get
$$P\Big( 0 < p_i < \frac{\omega_i}{\pi} \frac{m}{t},~0<p_j < \frac{\omega_j}{\pi} \frac{m}{\delta_{t}t}
\Big)=o\big(\frac1{t^{1+\gamma}}\big).$$
2) Farlie-Gumbel-Morgenstern (FGM) Copula. The copula function is
 $$\mathbb{C}(u_i,v_j) = u_i v_j\{1+\theta(1-u_i)(1-v_j)\},~ \theta \in [-1,1],~ 1\leq i\neq j\leq m. $$
When $m$ is divergent and $m=o(t^{\gamma/2})$, let $\delta_{t}=t^{\gamma-1}$, where $\gamma \in (0,1]$.
From (3), we have
$$P\Big( 0 < p_i < \frac{\omega_i}{\pi} \frac{m}{t},~0<p_j < \frac{\omega_j}{\pi} \frac{m}{\delta_{t}t}
\Big)=o\big(\frac1{t^{1+\gamma}}\big),$$
and from (4), we can get
$$P\Big( 0 < p_i < \frac{\omega_i}{\pi} \frac{m}{t},~1-\frac{\omega_j}{\pi} \frac{m}{\delta_{t}t} < p_j < 1
\Big)=o\big(\frac1{t^{1+\gamma}}\big).$$
3) Cuadras-Aug\'{e} Copula. The copula function is
$$\mathbb{C}(u_i,v_j) = \big(\min{ (u_i, v_j)}\big)^{\theta}(u_i v_j)^{1-\theta}, \theta \in [0,1], 1\leq i\neq j\leq m.$$
When $m$ is divergent and $m=o(t^{\gamma/2})$, let $\delta_{t}=t^{\gamma-1}$, where $\gamma \in (0,1]$.
From (5), we can get
$$P\Big( 0 < p_i < \frac{\omega_i}{\pi} \frac{m}{t},~0<p_j < \frac{\omega_j}{\pi} \frac{m}{\delta_{t}t}
\Big)=o\big(\frac1{t^{1+\gamma}}\big),$$
and from (6), we have
$$P\Big( 0 < p_i < \frac{\omega_i}{\pi} \frac{m}{t},~1-\frac{\omega_j}{\pi} \frac{m}{\delta_{t}t} < p_j < 1
\Big)=o\big(\frac1{t^{1+\gamma}}\big).$$
4) Normal Copula. The copula function is
$$\mathbb{C}(u_i,v_j) = \frac{1}{2\pi\sqrt{1-\rho_{ij}}}\int_{-\infty}^{\Phi^{-1}(u_i)}\int_{-\infty}^{\Phi^{-1}(v_j)}\exp\left(-\frac{x^2-2\rho_{ij}xy+y^2}{2(1-\rho^2_{ij})}\right)dxdy, 1\leq i\neq j\leq m.$$
When $m$ is divergent, $0\leq\rho_{ij}<\rho_{\max}<1$ and $m=o(t^{\gamma/2})$, we take $\delta_t t = t^{\beta}$ where ${(1+\rho_{max})}/{2} \leq \beta<1$, namely, $\delta_t={t^{\beta-1}} $, $\gamma={2}\beta/{(1+\rho_{\max})}-1(0<\gamma \leq 1)$. From (7), we have
$$P\Big( 0 < p_i < \frac{\omega_i}{\pi} \frac{m}{t},~0<p_j < \frac{\omega_j}{\pi} \frac{m}{\delta_{t}t}
\Big)=o\big(\frac1{t^{1+\gamma}}\big),$$
and from (8), we can get
$$P\Big( 0 < p_i < \frac{\omega_i}{\pi} \frac{m}{t},~0<p_j < \frac{\omega_j}{\pi} \frac{m}{\delta_{t}t}
\Big)=o\big(\frac1{t^{1+\gamma}}\big).$$
5) Ali-Mikhail-Haq (AMH) Copula.
(i) When $\theta \in [0,1]$, the copula function is $$\mathbb{C}(u_i,v_j) = \frac{u_iv_j}{1-\theta(1-u_i)(1-v_j)},~1\leq i\neq j\leq m.$$
When $m$ is divergent and $m=o(t^{\gamma/2})$, let $\delta_{t}=t^{\gamma-1}$, where $\gamma \in (0,1]$. From (9), we have
$$P\Big( 0 < p_i < \frac{\omega_i}{\pi} \frac{m}{t},~0<p_j < \frac{\omega_j}{\pi} \frac{m}{\delta_{t}t}
\Big)=o\big(\frac1{t^{1+\gamma}}\big),$$
and from (10), we can get
$$P\Big( 0 < p_i < \frac{\omega_i}{\pi} \frac{m}{t},~0<p_j < \frac{\omega_j}{\pi} \frac{m}{\delta_{t}t}
\Big)=o\big(\frac1{t^{1+\gamma}}\big).$$

\noindent(ii) When $\theta \in [-1,0]$, the copula function is
 $$\mathbb{C}(u_i,v_j) = \frac{u_iv_j}{1-\theta(1-u_i)(1-v_j)},~ 1\leq i\neq j\leq m. $$
When $m$ is divergent and $m=o(t^{\gamma/2})$, let $\delta_{t}=t^{\gamma-1}$, where $\gamma \in (0,1]$. From (11), we can get
$$P\Big( 0 < p_i < \frac{\omega_i}{\pi} \frac{m}{t},~0<p_j < \frac{\omega_j}{\pi} \frac{m}{\delta_{t}t}
\Big)=o\big(\frac1{t^{1+\gamma}}\big),$$
 from (12), we have
$$P\Big( 0 < p_i < \frac{\omega_i}{\pi} \frac{m}{t},~0<p_j < \frac{\omega_j}{\pi} \frac{m}{\delta_{t}t}
\Big)=o\big(\frac1{t^{1+\gamma}}\big).$$
6) Survival Copula: The copula function is
$$\mathbb{C}(u_i,v_j) = {u_iv_j}\exp{(-\theta \ln u_i \ln v_j)},~ \theta \in [0,1],~ 1\leq i\neq j\leq m. $$
When $m$ is divergent and $m=o(t^{\gamma/2})$, let $\delta_{t}=t^{\gamma-1}$, where $\gamma \in (0,1]$.
From (13), we have
$$P\Big( 0 < p_i < \frac{\omega_i}{\pi} \frac{m}{t},~0<p_j < \frac{\omega_j}{\pi} \frac{m}{\delta_{t}t}
\Big)=o\big(\frac1{t^{1+\gamma}}\big),$$
from (14), we can get
$$P\Big( 0 < p_i < \frac{\omega_i}{\pi} \frac{m}{t},~0<p_j < \frac{\omega_j}{\pi} \frac{m}{\delta_{t}t}
\Big)=o\big(\frac1{t^{1+\gamma}}\big).$$

\end{proof}

}

\begin{proof}[Proof of Corollary 1]
\noindent  Let $x$ be an arbitrary positive real number and
$$c_m(x) = [2-(\log{m})^{-1}] x + \Big(2\log {m}-[\log{(\log{m})} + \log(4\pi) - \log 4] +
\frac{\log{(\log{m})}+\log(4\pi)- \log 4}{2 \log{m}}\Big).$$
Let $Y_1, ...,Y_p$ be {\em iid} and follow $ N(0,1)$. Under (D3), we have to prove that
$$\left|P\left([\max\limits_{s=1, ...,m}(|Z_s|)]^2 < c_m(x)\right)-P\left([\max\limits_{s=1, ...,m}(|Y_s|)]^2 < c_m(x)\right)\right|\rightarrow 0,~~as ~~m \rightarrow  \infty . $$
Let $$\phi(x,y,r)=\frac{1}{2\pi\sqrt{1-r}}\exp\left(-\frac{x^2-2rxy+y^2}{2(1-r^2)}\right).$$ According to the inequality in Berman  \citep{r5}, 
 we have
\begin{align*}
&\left|P\left([\max\limits_{s=1, ...,m}(|Z_s|)]^2 < c_m(x)\right)-P\left([\max\limits_{s=1, ...,m}(|Y_s|)]^2 < c_m(x)\right)\right|\\
& \geq 2\sum_{1\leq q<s\leq m }|\rho_{j,k}|\phi(\sqrt{c_m(x)}~,\sqrt{c_m(x)}~;|\rho_{q,s}|)
+
2\sum_{1\leq q<s\leq m }|\rho_{j,k}|\phi(\sqrt{c_m(x)}~,-\sqrt{c_m(x)}~;|\rho_{q,s}|).
\end{align*}
Let $c_A$ be a generic finite positive constant. Note that $c_m(x) = 2 \log m - \log(\log m) + o_p(1)$, we have
\begin{align*}
2\sum_{1\leq q<s\leq m }|\rho_{q,s}|\phi(\sqrt{c_m(x)}~,\sqrt{c_m(x)}~;|\rho_{q,s}|)
&\leq
c_A\sum_{1\leq q<s\leq m }|\rho_{q,s}|e^{\frac{-c_m(x)}{1+|\rho_{q,s}|}}\\
&\leq
c_A~\log m \sum_{1\leq q<s\leq m }|\rho_{q,s}| m^{-\frac{2}{1+|\rho_{q,s}|}}.
\end{align*}
Using the technique in \citep{r4}, 
the left expression of the above inequation is
\begin{align*}
2\sum_{1\leq q<s\leq m}|\rho_{q,s}|&\phi(\sqrt{c_m(x)}~,-\sqrt{c_m(x)}~;|\rho_{q,s}|)\leq c_A\sum_{1\leq q<s\leq m }|\rho_{q,s}|e^{\frac{-c_m(x)}{1-|\rho_{q,s}|}}\\
&\leq c_A\sum_{1\leq q<s\leq m }|\rho_{q,s}|e^{-c_m(x)}\leq c_A (\log m)m^{-2}\sum_{i=1}^{m-1}(m-i)\delta_i,
\end{align*}
and the right one goes to 0 as $m\rightarrow \infty$.
Under (D3), $\delta_m(\log m)^{2+\beta}\rightarrow 0$, as $m \rightarrow \infty$, we have
$$
c_A (\log m)m^{-2}\sum_{i=1}^{m-1}[(m-i)\delta_i] \leq c_A (\log m)m^{-1}\sum_{i=1}^{m-1}\delta_i\rightarrow 0,
$$
which means
$$\left|P\left([\max\limits_{s=1, ...,m}(|Z_s|)]^2 < c_m(x)\right)-P\left([\max\limits_{s=1, ...,m}(|Y_s|)]^2 < c_m(x)\right)\right|\rightarrow 0,~~as ~~m \rightarrow  \infty . $$

\end{proof}

\begin{proof}[Proof of Theorem 5]
\noindent First we will decompose $$\text{CCT} = \sum\limits_{j  = 1}^{m}\omega_j\tan\Big[\Big(0.5-2\big(1-\Phi(|Z_j|)\big)\Big)\pi\Big]$$ to three portions.
  Assume that $\min\limits_{i=1, ...,m}\omega_i ={c_0}/{m}$, where the constant $c_0 > 0$ and according to the definition of the CCT, we can figure out the conclusion below
\begin{align*}
\text{CCT}
&= \sum\limits_{j \in \Omega}\omega_j\tan\Big[\Big(0.5-2\big(1-\Phi(|Z_j|)\big)\Big)\pi\Big] + \sum\limits_{j \in \Omega^c}\omega_j \tan\Big[\Big(0.5-2\big(1-\Phi(|Z_j|)\big)\Big)\pi\Big]\\
&\geq \sum\limits_{j \in \Omega^c}\omega_j \tan\Big[\Big(0.5-2\big(1-\Phi(|Z_j|)\big)\Big)\pi\Big] +  \frac{m^\eta-c_0}{m}\tan\Big[\Big(0.5-2\big(1-\Phi(\min\limits_{j \in \Omega}|Z_j|)\big)\Big)\pi\Big]\\
&+\frac{c_0}{m}\tan\Big[\Big(0.5-2\big(1-\Phi(\max\limits_{j \in \Omega}|Z_j|)\big)\Big)\pi\Big]\\
& \widehat=~ \text{CCT}_1 + \text{CCT}_2 + \text{CCT}_3
\end{align*}
According to {Theorem 1}, we can deduce that the CCT$_1$ approximates a standard Cauchy distribution as $m$ goes to infinity, which indicates that
  $$\text{CCT}_1 =O_p(1).$$
Based on Theorem 3 of Liu and Xie \citep{r22}, we can get
  $$\text{CCT}_2  \geq o_p(1).$$
As $m$ goes to infinity, we have
$$ \text{CCT} \geq \text{CCT}_3 +  O_p(1). $$
The power of  the CCT is
\begin{align*}
\beta_{\text{CCT}}
&=P\left(\text{CCT}>C_\alpha\right)\\
&\geq P\left(\frac{c_0}{m}\tan\Big[\Big(0.5-2\big(1-\Phi(\max\limits_{j \in \Omega}|Z_j|)\big)\Big)\pi\Big] +  O_p(1)>C_\alpha  \right)\\
&=P\left(\tan\Big[\Big(0.5-2\big(1-\Phi(\max\limits_{j \in \Omega}|Z_j|)\big)\Big)\pi\Big]  > \big(C_\alpha - O_p(1)\big)m/c_0 \right)\\
&=P\left(\max\limits_{s\in \Omega} | Z_j |> \Phi^{-1}\Big[1-\frac{1}{2\pi}\arctan\frac{c_0}{\big(C_\alpha - O_p(1)\big)m} \Big] \right)\\
&=P\left(\max\limits_{s\in \Omega} | Z_j |> (2\log m)^{1/2}-\frac{\log \log m +\log(4\pi)}{(8\log m)^{1/2}} +\frac{\log [2\pi\big(C_\alpha-O_p(1)\big)]}{(2\log m)^{1/2}} +O(\frac{1}{\log m})  \right),
\end{align*}
where the last equation is based on (iii) of Lemma \ref{le2} .
Second, we divide the power of MAX to two parts
\begin{align*}
\beta_{\text{MAX}}
&=P\left(\frac{\sqrt{\text{MAX}}-a_m - O\Big({(\log m)^{-1} }\Big)}{b_m}>G_\alpha\right)\\
&=P\left(\sqrt{\text{MAX}}>G_\alpha b_m + a_m \right)\\
&=P\left(\max\limits_{j=1, ...,m} | Z_j |>G_\alpha b_m + a_m +  O\Big({(\log m)^{-1} }\Big) \right)\\
&=P\left(\max\{\max\limits_{j\in \Omega}| Z_j |,\max\limits_{j\in \Omega^c} | Z_j |\}>G_\alpha b_m + a_m  + O\Big({(\log m)^{-1} }\Big) \right)\\
&\leq P\left(\max\limits_{j\in \Omega} | Z_j |>G_\alpha b_m + a_m + O\Big({(\log m)^{-1} }\Big) \right) + P\left(\max\limits_{j\in \Omega^c}| Z_j |>G_\alpha b_m + a_m + O\Big({(\log m)^{-1} }\Big)  \right)\\
& \widehat=~ \beta_{\text{MAX}_1} + \beta_{\text{MAX}_2},
\end{align*}
where $a_m=(2\log m)^{1/2}-{(\log \log m +4\pi -4)}{(8\log m)^{-1/2}}$ and $b_m=(2 \log m)^{-1/2}$.
Then, we will prove that $\beta_{\text{MAX}_2} =P\left(\max\limits_{j\in \Omega^c}| Z_j |>G_\alpha b_m + a_m + O\Big({(\log m)^{-1} }\Big)\right) = o(1)$.
By simple transformation, we have
\begin{align*}
\beta_{\text{MAX}_2}
&=P\left(\frac{\max\limits_{j\in \Omega^c}| Z_j |-a_{m^\star}}{b_{m^\star}}>\frac{G_\alpha b_m + a_m + O\left(\frac{1}{\log m }\right)-a_{m^\star}}{b_{m^\star}}\right),
\end{align*}
where $m^\star = m^{\kappa}$, $b_{m^\star}=(2 \log m^\star)^{-1/2}$ and $a_{m^\star}=(2\log m^\star)^{1/2}-{(\log \log m^\star +4\pi -4)}{(8\log m^\star)^{-1/2}}$.
According to {Lemma \ref{le1}}, we have
\begin{align*}
\frac{a_m-a_{m^\star}}{b_{m^\star}}
&\geq 2[\sqrt{\log(m)\log(m^\star)}-\log(m^\star)].
\end{align*}
Due to $1>\kappa>0$, as both $m^{\star}$ and $m$ go to $\infty$, we have
\begin{align*}
\frac{a_m-a_{m^\star}}{b_{m^\star}}
&\geq 2\left[\sqrt{{\frac{1}{\kappa}}\log\left({m^\star}\right)\log(m^\star)}-\log(m^\star)\right]=2\left[\big(\sqrt{{\frac{1}{\kappa}}-1} \big)\log(m^\star)\right]\rightarrow \infty
\end{align*}
and
\begin{align*}
\frac{G_\alpha b_m }{b_{m^\star}}={{( \log m^\star)^{1/2}} \over {(\log m)^{1/2}}}\rightarrow 0,
\end{align*}
which means that
 \begin{align*}
 P\left(\max\limits_{j\in \Omega^c}| Z_j |>G_\alpha b_m + a_m + O\Big({(\log m)^{-1} }\Big)  \right) &\rightarrow 1-\exp\big(\exp(-\infty)\big) \text{, as }m^{\star} \text{ and } m \rightarrow \infty\\
&=o(1).
\end{align*}
Therefore, we have
$$\beta_{\text{MAX}} \leq  \beta_{\text{MAX}_1} + o(1).
$$
where $$\beta_{\text{MAX}_1} = P\left(\max\limits_{j\in \Omega} | Z_j |> (2\log m)^{1/2}-\frac{\log \log m +\log(4\pi) -\log(4)}{(8\log m)^{1/2}} + \frac{G_\alpha}{(2\log m)^{1/2}} + O\Big({(\log m)^{-1} }\Big) \right).$$
According to the above inequalities, as $m$ goes infinity, we can draw the conclusion
$$ \beta_{\text{MAX}}  \leq \beta_{\text{CCT}} + o(1) .$$

\end{proof}

\section*{Simulation results for Models 2 and 4}

Figures S1 and S2 display the tail probabilities of $P(\text{SCD}>t)$ (dotted line)  and $P(\text{CCT}>t)$ (solid line)  for Models 2 and 4, respectively, where  $P(\text{CCT}>t)$ is calculated based on 500,000 Monte Carlo samples.
Both figures indicate  that the standard Cauchy distribution can well approximate the tail probability of the CCT since both line are almost in coincidence.

\begin{figure}[htbp]
\includegraphics[width=15cm,height=12cm]{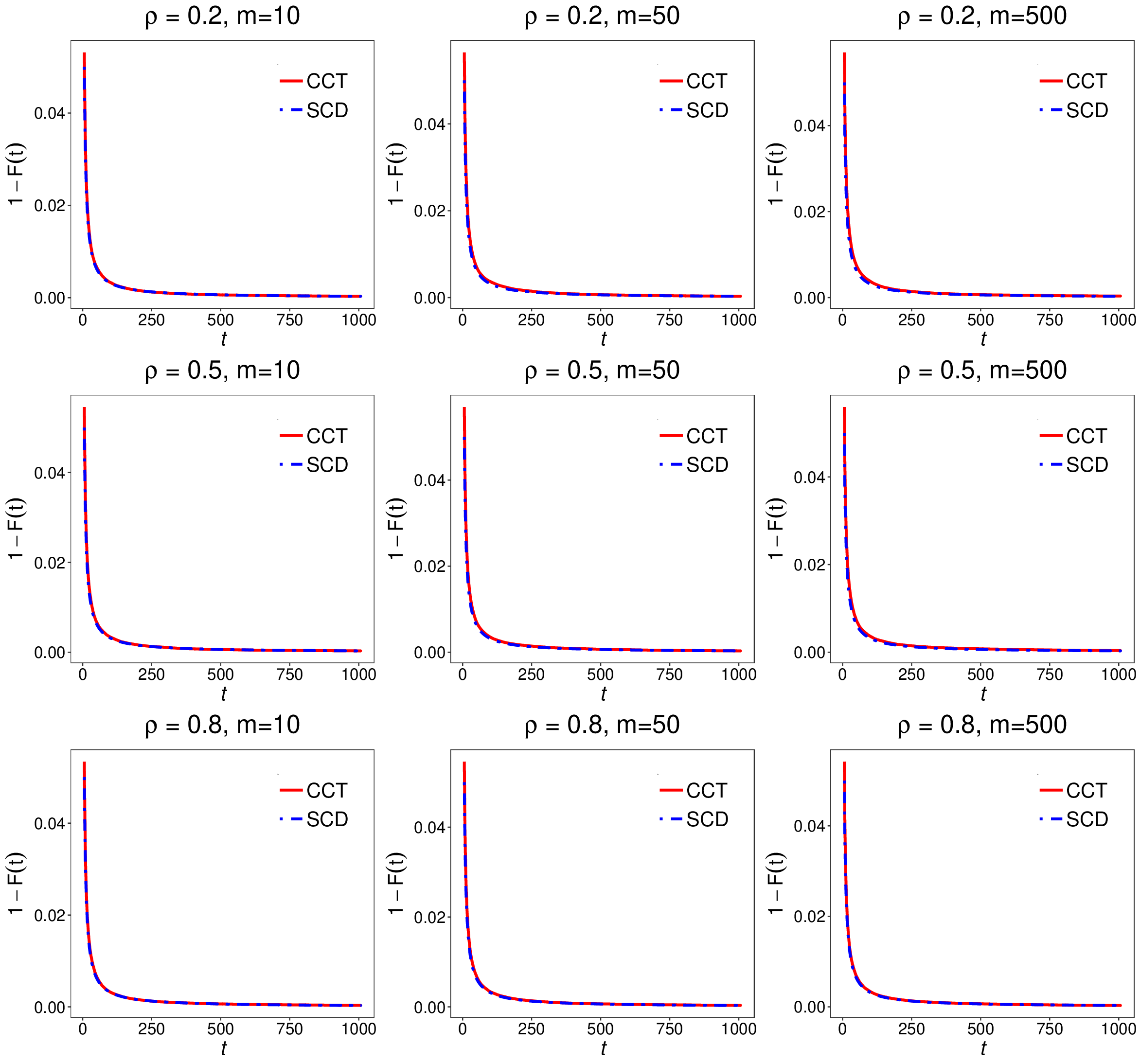}\\
{\small {\bf Figure S1.} The tail probability of the CCT (red solid line) and the SCD (blue dotted line), where test statistics $Z_{i}, i=1, ...,m$ are generated from
an $m$-variate $t$ distribution, with parameters given in Model 2. The vertical axis is the tail probability $1-F(t)$ for a cumulative distribution function $F$.}
\end{figure}

\begin{figure}[htbp]
\includegraphics[width=15cm,height=12cm]{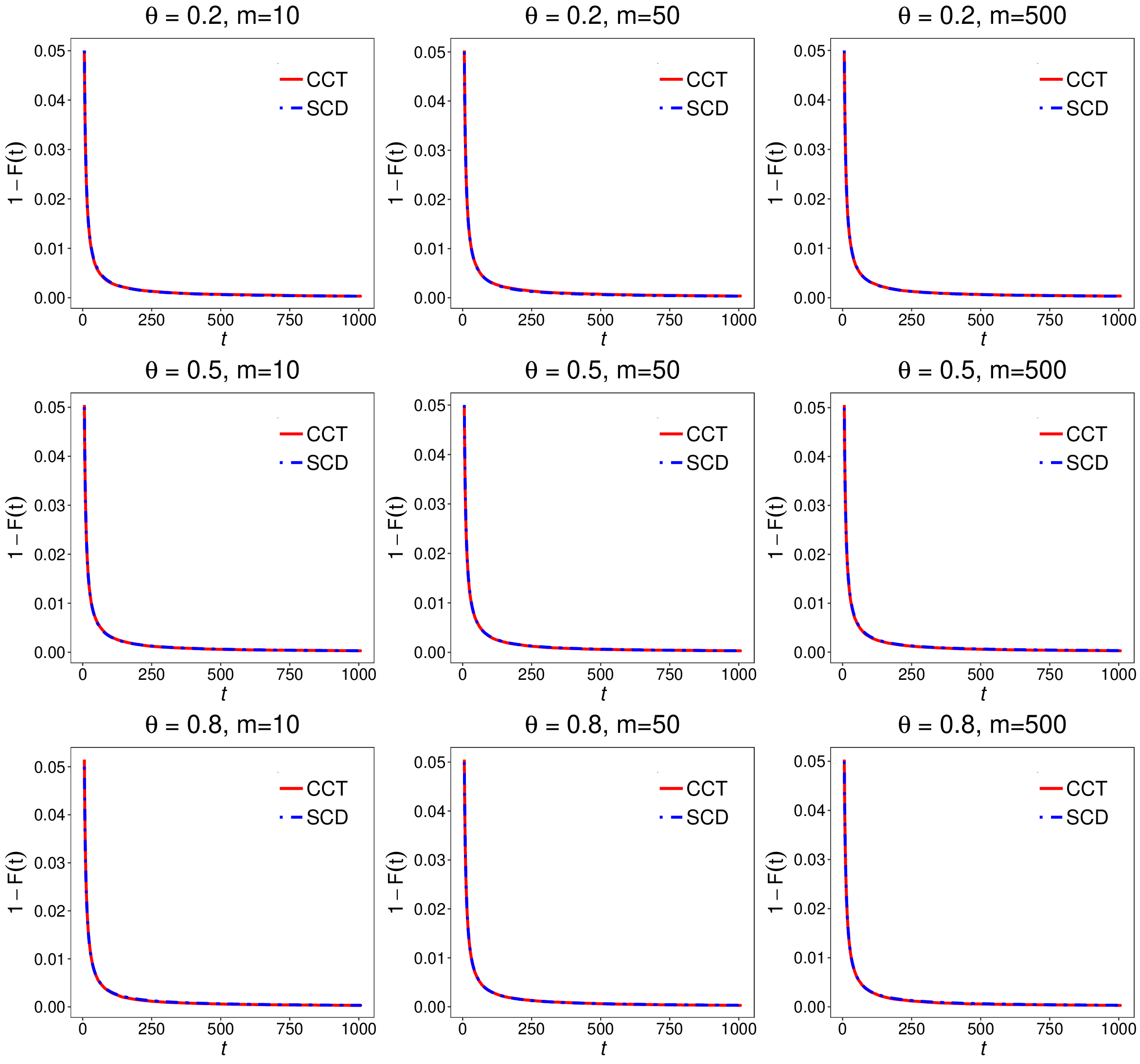}\\
{\small {\bf Figure S2.} The tail probability of the CCT (red solid line) and the SCD (blue dotted line), where test statistics $Z_{i}, i=1, ...,m$ are generated from
an $m$-variate $t$ distribution, with parameters given in Model 4. The  vertical axis is the tail probability $1-F(t)$ for a cumulative distribution function $F$.}
\end{figure}

\newpage

\section*{R codes}

\begin{lstlisting}[language=R]
##########-----------  Calculate CCT's P-values -------------------
CCT.cal.p <- function(pvalues)
{
    testvalue = mean(tan((0.5-pvalues)*pi))
    CCT.p = 1 - pcauchy(testvalue)
    CCT.p
}
##########----------        Model       1   -----   P = 10   ------
ourcauchyt1 = function(dataz,x,sigma)
{
  n=nrow(dataz)
  testvalue=rep(0,n)
  varmy = diag(sigma)
  for (i in 1:n)
  {
    testvalue[i]=mean(tan((0.5-2*(1-pt(abs(dataz[i,]/(sqrt(varmy))),
	df =10)))*pi))
  }
  nx=length(x)
  y2=rep(0,nx)
  for (j in 1:nx) {
    y2[j] = sum(testvalue>x[j])
  }
  y2 = y2/n
  y2
}
getcov1 = function(p,pho)
{
  c=matrix(NA,p,p)
  for (i in 1:p)
  {
    for(j in 1:p)
    {
      c[i,j] = pho^(abs(i-j))
    }
  }
  c
}
Spikedmodel1 = function(m,d)
{
  vectors=eigen(getcov1(m,0.8))$vectors
  lamda=rep(0,m)
  for (i in 1:d) {
    lamda[i]=m/(3^i)
  }
  for (i in (d+1):m) {
    lamda[i]=1
  }
  skipematrix=matrix(0,m,m)
  for (i in 1:m) {
    skipematrix=vectors[i,]%*%t(vectors[i,])*(lamda[i])+skipematrix
  }
  skipematrix
}
library(mvtnorm)
library(ggplot2)
library(lubridate)
library(ggpubr)
n = 500000
q = qcauchy(0.95,location = 0,scale = 1)
x = seq(q, q+1000, by = 0.1)
y1 = 1-pcauchy(x, location= 0, scale = 1)
m = 10
sigmamy = Spikedmodel1(m,4)
dataz = rmvt(n, sigma = sigmamy,df = 10 , delta = rep(0,m))
y2 = ourcauchyt1(dataz,x,sigmamy)
t  = x
Test = rep(c('SCD','CCT'),each = length(x))
Power = c(y1,y2)
df = data.frame(t = t, Test = Test, Power = Power)
p1 = ggplot(data = df, mapping = aes(x = t, y = Power,labels = Test,
  linetype = Test, colour = Test,  fill = Test))+
  scale_fill_manual(values=c("CCT" = "red", "SCD" = "blue"))+
  scale_colour_manual(values=c("CCT" = "red", "SCD" = "blue" ))+
  scale_linetype_manual(values = c("CCT" = "solid", "SCD" = "dotdash"))+
  geom_line(lwd = 1.3)+
  theme_bw()+
  theme(panel.grid.major = element_blank(),
        panel.grid.minor = element_blank(),
        legend.position = c(0.8,0.8),legend.key.size = unit(1,'cm')) +
  labs(title=expression(paste("m=10, d=4")))  +
  theme(plot.title = element_text(hjust = 0.5, size = 25))+
  theme(axis.text.x = element_text( color="black", size=16))+
  theme(axis.text.y = element_text( color="black", size=16))+
  theme(axis.title.x = element_text( face = "italic",
                                     color="black", size=20))+
  theme(axis.title.y = element_text( face = "italic",
                                     color="black", size=20))+
  guides(fill = guide_legend(title = waiver(),
                             title.theme = element_text(size=1),
                             label.theme = element_text(size = 20)))
  labs(y = quote(1-F(t)))
p1 = p1+ theme(legend.title=element_blank())
#--- Similarly, we can get pictures p2, p3, p4, p5, p6, p7, p8, p9
#--- with m = 50, 500 and d = 5, 6

##########----------       Air Quality Analysis -----------
rm(list=ls())
library(mvtnorm)
library(GBJ)
data0 = as.matrix(read.csv(file="data.csv", header=F, stringsAsFactors=F))
X = data0[-1,c(3:12)]
X = X[1:9357,]
X = matrix(as.numeric(as.character(X)),9357,10)
Y = data0[-1,c(13:15)]
Y = Y[1:9357,]
Y = matrix(as.numeric(as.character(Y)),9357,3)
n = nrow(X)
m = ncol(X)
miss = union(which(Y[,1] == -200),which(Y[,2] == -200))
miss = union(miss,which(Y[,3] == -200))
miss = union(miss,which(X[,1] == -200))
miss = union(miss,which(X[,2] == -200))
miss = union(miss,which(X[,3] == -200))
miss = union(miss,which(X[,4] == -200))
miss = union(miss,which(X[,5] == -200))
miss = union(miss,which(X[,6] == -200))
miss = union(miss,which(X[,7] == -200))
miss = union(miss,which(X[,8] == -200))
miss = union(miss,which(X[,9] == -200))
miss = union(miss,which(X[,10] == -200))
Y = Y[-miss,]
X = X[-miss,]
m = ncol(X)
pvalue = rep(0,m)
for (i in 1:m) {
  fit <- lm(X[,i]~ Y)
  pvalue[i]   = anova(fit)$`Pr(>F)`[1]
}
MIN.obs= min(pvalue)
CCT.obs= mean(tan((0.5-pvalue)*pi))
CCTpvalue=1-pcauchy(CCT.obs)
pvalue.null = matrix(0,10^6,m)
for (i in 1:10^6) {
  for (j in 1:m) {
    fit <- lm(X[sample(s,s),j]~ Y)
    pvalue.null[i,j]   = anova(fit)$`Pr(>F)`[1]
  }
}
 null.p = apply(pvalue.null, 1, min)
 MINpvalue = sum(null.p < MIN.obs)/10^6
 cormatrix = matrix(0,10^{6},m)
 zvalue = rep(0,m)
 for (i in 1:10^{6}) {
   for (j in 1:m) {
     fit <- lm(X[sample(s,s),j]~ Y)
     cormatrix[i,j]   = anova(fit)$`F value`[1]
   }
 }
 mymean = apply(cormatrix, 2, mean)
 mysd =  apply(cormatrix, 2, sd)
 cormatrix = scale(cormatrix)
 mycor = cor(cormatrix)
   for (j in 1:m) {
     fit <- lm(X[,j]~ Y)
     zvalue[j]=(anova(fit)$`F value`[1]-mymean[j])/mysd[j]
   }
 GHCpvalue = GHC(zvalue,cor_mat = mycor)$GHC_pvalue
 GBJpvalue = GBJ(zvalue,cor_mat = mycor)$GBJ_pvalue
 c(CCTpvalue, MINpvalue, GHCpvalue, GBJpvalue)
\end{lstlisting}

\end{document}